\documentclass[final,journal]{IEEEtran}
\usepackage{amsmath,amsfonts}
\usepackage{algorithmic}
\usepackage{algorithm}
\usepackage{array}
\usepackage{stfloats}
\usepackage{url}
\usepackage{verbatim}
\usepackage{graphicx}
\usepackage{cite}

\usepackage{xcolor}         
\definecolor{mycyan}{RGB}{0,204,204}

\usepackage{microtype}
\usepackage{graphicx}
\usepackage{booktabs} 
\usepackage{cuted}
\usepackage{caption}
\usepackage{float} 
\usepackage{subfigure}

\usepackage{kz_style}

\begin{document}

\title{
Distributed Policy
Gradient 
for Linear \\Quadratic  Networked Control 
with Limited\\ Communication Range}

\author{Yuzi Yan,~\IEEEmembership{Student Member,~IEEE,} 
        and Yuan Shen,~\IEEEmembership{Senior Member,~IEEE,}
\thanks{The authors are with the Department of Electronic Engineering, and Beijing National Research Center for Information Science and Technology, Tsinghua University, Beijing 100084, China (e-mail: yyz21@mails.tsinghua.edu.cn; shenyuan\_ee@tsinghua.edu.cn).}
}

\markboth{Journal of \LaTeX\ Class Files,~Vol.~14, No.~8, August~2021}%
{Shell \MakeLowercase{\textit{et al.}}: A Sample Article Using IEEEtran.cls for IEEE Journals}


\maketitle
\begin{abstract}
This paper proposes a scalable distributed policy gradient method and proves its convergence to near-optimal solution in multi-agent linear quadratic networked systems. The agents engage within a specified network under local communication constraints, implying that each agent can only exchange information with a limited number of neighboring agents. On the underlying graph of the network, each agent implements its control input depending on its nearby neighbors' states in the linear quadratic control setting. We show that it is possible to approximate the exact gradient only using local information. Compared with the centralized optimal controller, the performance gap decreases to zero exponentially as the communication and control ranges increase. We also demonstrate how increasing the communication range enhances system stability in the gradient descent process, thereby elucidating a critical trade-off. The simulation results verify our theoretical findings.
\end{abstract}


\section{Introduction}
\label{sec:introduction}
\IEEEPARstart{R}{ecent} years have seen significant advances in optimizing stochastic dynamical multi-agent networked systems~\cite{Berahas2021on}, such as communication~\cite{chang2015multi},  cooperative detection~\cite{gu2020Cooperative,gu2021Quasi}, games~\cite{silver2016mastering}, formation~\cite{yan2022relative,pan2022flexible,yan2023approximation}, localization~\cite{shen2023theoretical}, smart grid~\cite{chen2022reinforcement}, etc. Compared to single-agent control, the multi-agent setting presents additional challenges due to communication limitations between agents, with scalability being a primary concern~\cite{zhang2018fully}. Although some scalable reinforcement learning (RL) methods have been proposed, there needs to be more theoretical understanding regarding their optimality and efficiency~\cite{zhang2021multi}. On the other hand, control theory provides us with a rich body of tools with provable guarantees~\cite{tian2022can,zhang2020stability,zhang2019policy,fazel2018global}. Specifically, linear quadratic regulator (LQR) is one of the most well-studied optimal control problems, which considers optimal state feedback control for a linear dynamical system~\cite{bakule2008decentralized}. Networked linear quadratic control has many applications, including transportation~\cite{bazzan2009opportunities}, power grid~\cite{pipattanasomporn2009multi}, etc.

Recent studies have shown that within a centralized LQR framework, it is feasible to derive the optimal controller using policy gradient descent techniques~\cite{fazel2018global,bu2019lqr}. However, to the best of our knowledge, the question of whether one can acquire a near-optimal decentralized controller through a scalable and distributed policy gradient descent approach still remains unresolved. There are many challenging problems at both theoretical and practical levels:
First, it is a non-convex optimization problem, where in general the gradient descent does not converge to the global optima in the limit.
Second, each agent must accurately approximate the exact global gradient using limited communication, which is critical to achieving scalability for excessively large networked multi-agent systems. 
Third, the dynamical system may suffer from the risk of becoming unstable during the distributed policy gradient process. It is necessary to provide a stability guarantee.

This paper first provides theoretical results about the problems mentioned above, encouraged by the recent success in the scalable optimization algorithm~\cite{li2021distributed}. With a few mild assumptions about the state transition matrix of the networked system in the Markov Decision Process (MDP) setting, every agent can approximate an accurate localized policy gradient in a distributed fashion. The proposed distributed policy gradient method can attain near-optimal solutions with a gap that is exponentially small relative to the communication range $\kappa$ and the control range $r$. By carefully selecting the step size $\eta$, the system’s stability can be ensured. Numerical results for several representative cases are also presented to illustrate these findings.

Our main contributions are summarized as follows.

\begin{itemize}
  \item [$\bullet$]
  \emph{Localized Gradient Approximation.}
  We provide conditions for accurate gradient approximation with limited local communication in a networked system.
  This work analyzes the conditions for the establishment of Exponential Decay Property in a networked LQR setting, which is a crucial concept in many scalable multi-agent RL algorithms. In specific, for the Exponential Decay Property to hold, we need some mild assumptions on the state transition matrix, which indicate that the interaction between agents should be weak enough or the network's connectivity should be low enough. 
  \item [$\bullet$]
  \emph{Stability Guarantee in the Gradient Descent Process.} We provide a stability guarantee during the gradient process. We quantify the two factors that may cause the system to become unstable: 1) the inappropriate step size, 2) the error brought by the localized gradient approximation. The guarantee is necessary because the stability of the system is a prerequisite for the performance analysis.
  \item [$\bullet$]
  \emph{Near Optimal Performance for the Distributed Policy Gradient.}  We show that the decentralized controller obtained through the distributed policy gradient achieves near-global optimality, where each agent's communication is restricted within its $\kappa$-hop local neighborhood and takes control actions based on its $r$-hop local observations. The optimality gap compared with the centralized optimal controller is proven to be exponentially small in $\kappa$ and $r$. 
\end{itemize}

\noindent\textbf{Notation.} Throughout the paper, we use $\Vert \cdot \Vert$ to denote the $l_2$ norm of a vector and the induced $l_2$ norm of a matrix. $\Vert \cdot \Vert_F$ denotes its Frobenius norm. $\Abf^\Tb$, $\rho(\Abf)$, and $\tr(\Abf)$ represent the transpose, spectral radius, and trace of the matrix $\Abf$. $\Abf \succeq \Bbf$ for two symmetric matrices refers to the positive semi-definiteness of their difference $\Abf-
\Bbf$. We let $\sigma_i(\Abf)$ represent the eigenvalues of a square matrix $\Abf$, which are indexed in increasing order for their real parts, i.e.,
$
\textbf{Re} (\sigma_1(\Abf)) \leq \cdots \leq \textbf{Re} (\sigma_n(\Abf))
$.
If $\Abf$ is symmetric, the ordering becomes $\sigma_1(\Abf) \leq \cdots \leq \sigma_n(\Abf)$. Suppose that $\Abf$ is a partitioned matrix, $[\Abf]_{ij}$ denotes the sub-matrix of $\Abf$ where its row indexes correspond to the role indexes of agent $i$ and its column indexes correspond to the indexes of agent $j$. We use $\overline{[\Abf]}$ to denote $\max_{i,j} \Vert [\Abf]_{ij} \Vert$ and $\underline{[\Abf]}$ to denote $\min_{i,j} \Vert [\Abf]_{ij} \Vert$. Without further explanation, $\dist(\cdot, \cdot)$ refers to the graph distance between two agents in the system throughout the paper. $\dist(i,i)$ is defined to be $0$.

For integer $\kappa > 1$, $\cN_i^\kappa$ denotes the $\kappa$-hop neighborhood of $i$, i.e., the agents whose graph distance to $i$ is less than or equal to $\kappa$, including $i$ itself. $\cN_{-i}^\kappa$ denotes the set including all the agents that are outside $\kappa$-hop neighborhood of $i$. On an undirected graph $\cG=(\cN, \cE)$, given a fixed $\kappa$, we use $N_{\cG}^{\kappa}$ to denote the number of the agents in the biggest $\kappa$-hop neighborhood, i.e., $N_{\cG}^{\kappa} = \max_i |\cN_i^\kappa|$.

\section{Related Work}
\label{section:related workds}
We briefly discuss related works in areas such as optimization in networked systems, learning-based control, and scalable multi-agent reinforcement learning. These works significantly inspired this paper. A comparative analysis is then presented to highlight the distinctions between our work and these prior studies.
\subsection{Learning-based Control} 
In the context of LQR, the classical model-based approach dates back to subspace-based system identification~\cite{ljung1987theory}. More recently, much progress has been made in algorithm design and nonasymptotic analysis in both centralized and networked LQR optimization. 
 
Several studies, including those by~\cite{dean2020sample,tu2018least,simchowitz2018learning}, have addressed the issue from the perspective of a single agent. In \cite{fazel2018global}, the authors have provided proof of global optimality in the centralized LQR setting, specifically concerning deterministic policies and non-noisy transitions. Notably, their paper introduces a model-free algorithm based on zero-order optimization. Furthermore, \cite{yang2019provably} has demonstrated that an actor-critic algorithm for the single-agent LQR achieves optimal performance with a linear convergence rate. Global convergence of policy optimization methods for other control problems has also been studied lately \cite{zhang2019policy,zhang2021policy,guo2022global}, see \cite{hu2023toward} for an overview of the recent results on policy optimization for control.

Some other works~\cite{rotkowitz2005characterization,gagrani2018thompson} try to address the issue in a distributed fashion, where centralized methods face challenges related to scalability and communication overhead in extensive networked systems. In~\cite{bu2019lqr}, the authors examine LQR feedback synthesis with a sparsity pattern and guarantee a sublinear rate of convergence to a first-order stationary point. In~\cite{li2021distributed}, a zero-order distributed policy optimization algorithm is proposed for networked LQR. The algorithm is guaranteed to approach the stationary point. In~\cite{shin2022near}, the performance between the truncated distributed controller and the optimal controller is proved to be exponentially small  with some additional assumptions. In~\cite{zhang2022optimal}, the authors show that the optimal LQR state feedback gain is ``quasi''-SED (spatially-exponential decaying) with system matrices being SED between nodes in the network, suggesting that distributed controllers can achieve near-optimal performance for SED system. 

\subsection{Multi-agent Reinforcement Learning} 
Our problem can also be viewed as a special case of the cooperative setting in multi-agent RL (MARL), i.e., the Decentralized Partially Observable Markov Decision Process (Dec-POMDP). This intricate framework for multi-agent systems is characterized by limited agent observability. In a Dec-POMDP, multiple agents cooperate under the constraint of not having a complete perspective of the environment or the full actions of their peers~\cite{bernstein2002complexity}. 

In response to the challenges posed by Dec-POMDPs, numerous MARL algorithms have been developed. The policy gradient approach, within this space, emerges as a popular choice~\cite{sutton1999policy,silver2014deterministic}. In ~\cite{qu2019exploiting,qu2020average,qu2020scalable}, the authors introduced a Scalable Actor-Critic methodology, proficient at discerning a near-optimal localized policy in tabular RL contexts. In~\cite{alfano2021dimension}, the authors crafted a scalable algorithm rooted in the NPG framework, demonstrating its convergence to the globally optimal policy. Their analysis further illustrates that the performance disparity diminishes to zero at an exponentially rapid rate, contingent upon the communication range. In~\cite{zhang2022global}, the authors postulated a Localized Policy Iteration (LPI) algorithm, empirically validated to adeptly acquire a policy that is nearly globally optimal, leveraging solely localized data within a networked system. Their findings also spotlight that the optimality gap recedes polynomially, governed by the communication range parameter, $\kappa$.

\subsection{Comparison with Prior Work}
The results most pertinent to our study include those in \cite{qu2020average}, \cite{zhang2022optimal}, and \cite{shin2022near}. Here we aim to delineate the distinctions between our results and theirs. In~\cite{qu2020average}, the authors demonstrate that the Exponential Decay Property holds for MDPs with average rewards in the tabular case. In contrast, our study delves into continuous spaces under the LQR setting.  In~\cite{zhang2022optimal,shin2022near}, the authors elucidate the spatially exponential decaying attribute of the optimal controller concerning the networked LQR problem. Our perspective, however, concentrates on \emph{finding}  the near-optimal solution within the projected space. Notably, our contribution provides a distributed, scalable policy gradient technique to unearth this solution. We also assert its stability, a feature not showcased in the aforementioned studies.

\section{Preliminaries and Background}
\label{section:preliminary}

\subsection{Problem Formulation}
\label{subsection:problemformulation}
\begin{figure}
    \centering
    \includegraphics[width=0.8\linewidth]{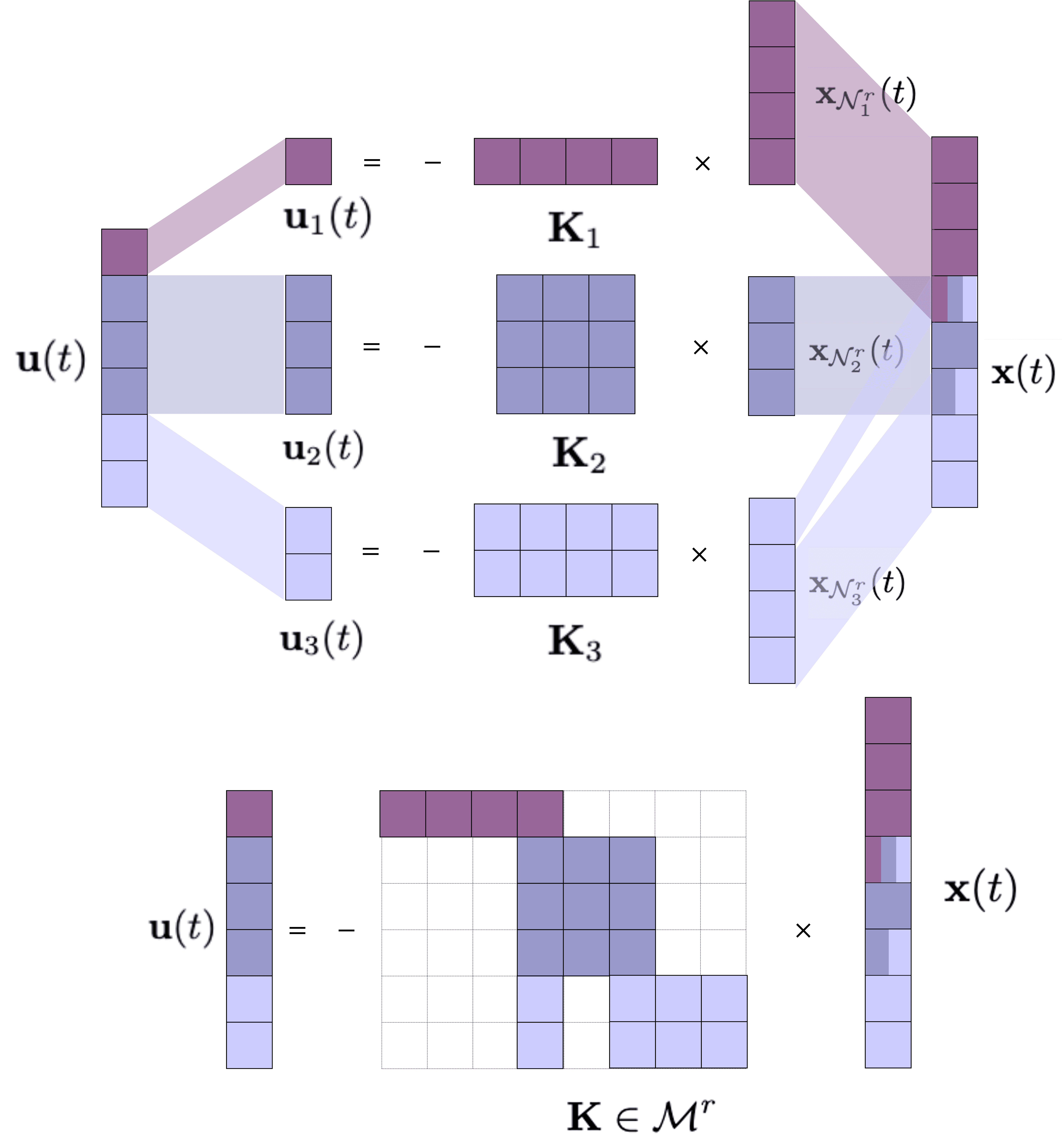}
    \caption{\small An illustrative diagram for a networked system. The top figure illustrates the local control inputs, local controllers, and local observations; and the bottom figure provides a global perspective where the optimization objective controller $\Kbf \in \cM^r$. Note that we omit the noise added to the control output for simplicity.}
    \label{fig:figure_diagram1}
\end{figure}

\begin{figure*}
    \centering
    \includegraphics[width=0.8\linewidth]{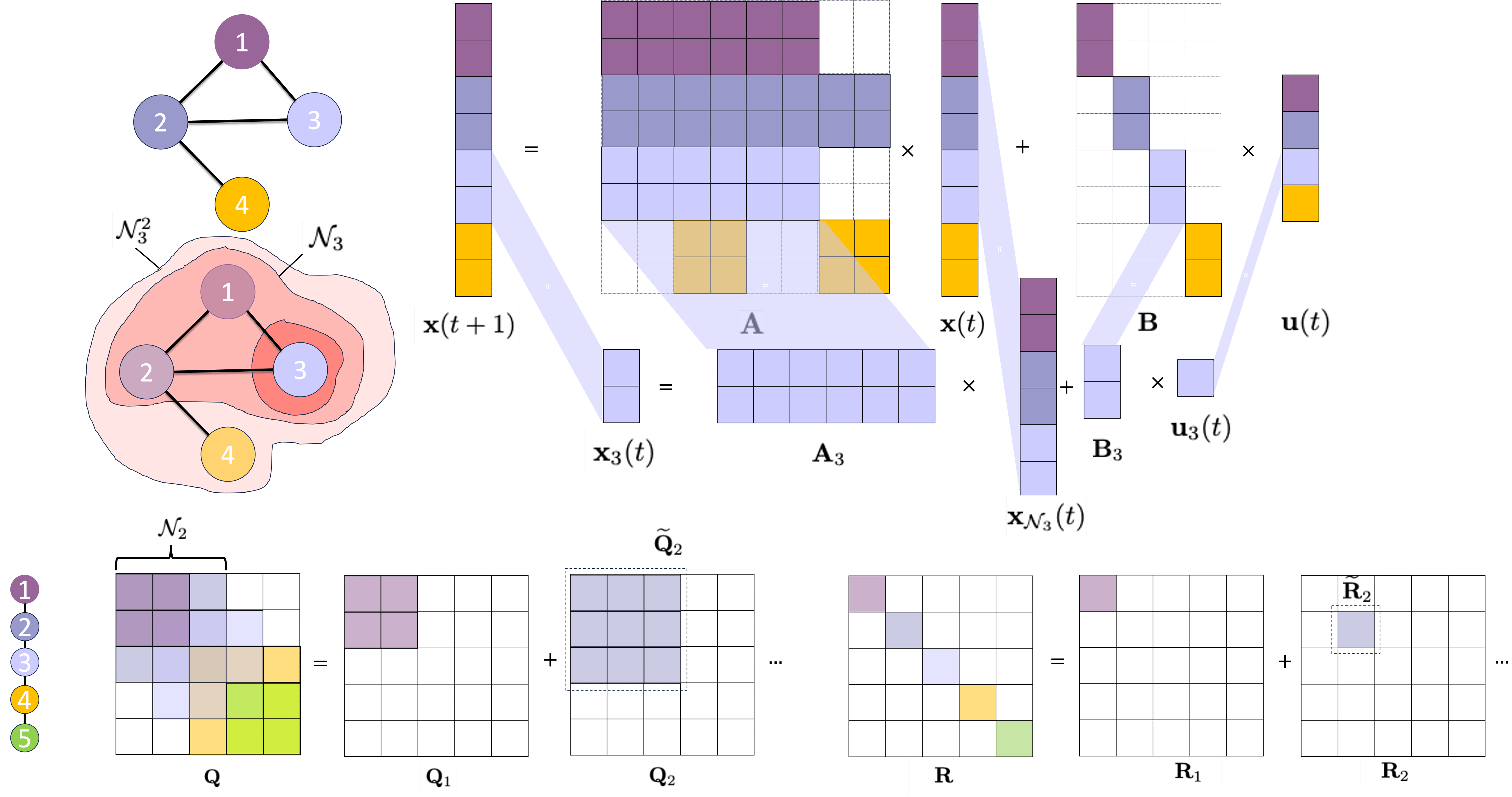}
    \caption{\small A diagram to show the spatial structures of the system matrices, including $\Abf$, $\Bbf$, $\Qbf$ and $\Rbf$.}\label{fig:figure_diagram2}
\end{figure*}
First, we introduce the basic setup of networked LQR control. Suppose there are $n$ agents jointly controlling a discrete-time linear dynamic system. The state space $\cX$ and action space $\cU$ are $\RR^{d_x}$ and $\RR^{d_u}$. The transition dynamics and the quadratic state cost $c(t)$ of the whole system are given by: 
\mequa
 \xbf(t+1) & = \Abf \xbf(t) + \Bbf \ubf (t) + \wbf(t), \\
 c(t) & =\xbf(t)^{\Tb} \Qbf \xbf(t) + \ubf(t)^{\Tb} \Rbf \ubf(t),
\mequa
where $\xbf(t)\in \RR^{d_x}$ and $\ubf(t)\in \RR^{d_u}$, $\wbf(t) \sim N(\mathbf{0}, \bPhif)$ denotes random noise that is i.i.d. for each time step $t$. Each agent $i$ is associated with a local state $\xbf_i(t) \in \RR^{d_x^i}$ and a control input $\ubf_i(t)\in \RR^{d_u^i}$, which jointly constitute the global state $\xbf(t) = [\xbf_1(t)^\Tb, \xbf_2(t)^\Tb, ..., \xbf_n(t)^\Tb]^\Tb$ and the global control input $\ubf(t) = [\ubf_1(t)^\Tb, \ubf_2(t)^\Tb, ..., \ubf_n(t)^\Tb]^\Tb$. Note that $d_x = \sum_{i=1}^n d_x^i$ and $d_u = \sum_{i=1}^n d_u^i$.

Next, we illustrate the structural properties of the network and demonstrate how these properties confer local structural characteristics to the optimization. We consider the case that there exists a network of $n$ agents with an underlying undirected graph $\cG=(\cN, \cE)$, where $\cN=\{ 1,2, \cdots, n\}$ is the set of agents and $\cE \subseteq \cN \times \cN$ is the set of edges. As a common setting in many practical scenarios~\cite{qu2019exploiting, qu2020average, qu2020scalable}, the next state of each agent $i$, denoted as $\xbf_i(t+1)$, along with the quadratic local cost $c_i(t)$, is dependent solely on the state of its 1-hop neighbors (denoted by $\cN_i$) and its own control input at time step $t$. We can use a series of smaller matrices $\{\Abf_i\}_{i=1}^{n}$ and $\{\Bbf_i\}_{i=1}^{n}$ to describe the local transition process. The transition dynamics of the linear system and the local cost could be rewritten as
\mequa
\xbf_i(t+1) &= \Abf_i \xbf_{\cN_i}(t) + \Bbf_i \ubf_i(t)+\wbf_i(t), \\
c_{i}(t) &= \xbf_{\cN_i}(t)^\Tb \tilde{\Qbf}_i \xbf_{\cN_i}(t) + \ubf_{i}(t)^\Tb \tilde{\Rbf}_i \ubf_{i}(t),
\mequa
where $\xbf_{\cN_i}$ is the concatenated state vector of agent $i$'s 1-hop neighbors. The global cost $c(t)=\frac{1}{n}\sum_{i=1}^{n}c_i(t)$.  $\tilde{\Qbf}_i$ and $\tilde{\Rbf}_i$ are positive definite matrices that parameterize the local quadratic costs. Note that if we pad $\tilde{\Qbf}_i$ and $\tilde{\Rbf}_i$ with 0, we can create $\Qbf_i$ and $\Rbf_i$,  and make them    the same size as $\Qbf$ and $\Rbf$. Then, the local state cost function $c_i(t)$ can be written as
\mequa
& c_i(t)=\xbf(t)^\Tb \Qbf_i \xbf(t) + \ubf(t)^\Tb \Rbf_i \ubf(t).
\mequa
We further define $\xbf_{\cN_i^\kappa}$ and $\ubf_{\cN_i^\kappa}$ to be the joint state and control vectors of agent $i$'s $\kappa$-hop neighbors, $\xbf_{\cN_{-i}^\kappa}$ and $\ubf_{\cN_{-i}^\kappa}$ to be the joint state and control vectors outside of the $\kappa$-hop neighborhood.\footnote{In the following, sometimes we use the form such as $\xbf=(\xbf_{\cN_i^\kappa},\xbf_{\cN_{-i}^{\kappa}})$. Note that we omit the rearrangement of the index of agents for notation simplicity.} Given that the global cost of the system $c(t)$ is the average of all the local costs $c_i(t)$, it follows that $\Qbf=\frac{1}{n} \sum_{i=1}^n \Qbf_i$ and $\Rbf=\frac{1}{n} \sum_{i=1}^n \Rbf_i$.

Next, we introduce the optimization objective. In the context of distributed optimization, the admissible local controllers are constrained to those that solely utilize the current state of $r$-hop neighbors, denoted as $\cN_i^r$, as input. $r$ reflects how distributed the  controllers could be. If $r$ reaches the diameter of the underlying graph, the controller will degenerate into a centralized one. Throughout the rest of this paper, we focus on the family of stochastic linear-Gaussian policies, i.e., $\pi_{\Kbf_i}(\cdot | \xbf_{\cN_i^r})=N(-\Kbf_i \xbf_{\cN_i^r}, \sigma_0^2 \bI_{d_u^i}), \Kbf_i\in \RR^{d_u^i \times \sum_{j\in \cN_i^r} d_x^j}$. We use $\Kbf$ to denote the control gain matrix of the entire system, $\pi_{\Kbf} (\cdot | \xbf)=N(-\Kbf \xbf, \sigma_0^2 \bI_{d_u}), \Kbf \in \RR^{d_u \times d_x}$. It is straightforward to see $\pi_{\Kbf} (\ubf | \xbf) = \prod_{i=1}^{n} \pi_{\Kbf_i}(\ubf_i | \xbf_{\cN_i^r})$.\footnote{We sometimes use $\left\{ \Kbf_1, \Kbf_2, ..., \Kbf_n \right\}$ to denote the controllers from the local perspective and note that it is equivalent to $\Kbf$ from the global perspective. We can combine all the $\Kbf_i$ and pad it with 0 to a $\RR^{d_u \times d_x}$ matrix to get $\Kbf$.}

\begin{definition}
\label{def:Mr}
For any integer $\kappa>1$, for a $n$-by-$n$ block matrix $\Mbf$, we define
\mequa
\cM^{\kappa}= \{ \Mbf: [\Mbf]_{ij}=0 \;\; \text{i.f.f.} \;\; \dist(i,j)>\kappa, \forall i,j \},
\mequa
and for the matrices related to a certain agent $i$, we define
\mequa
\cM_i^{\kappa}=\{ \Mbf : [\Mbf]_{ij}=0 \;\; \text{i.f.f.} \;\; \dist(i,j)>\kappa, \forall j\}.
\mequa
\end{definition}

It is direct to see that $\Abf$, $\Bbf$, $\Qbf$, $\Rbf$ and $\Kbf$ are all $n$-by-$n$ block matrices and $\Abf, \Qbf\in \cM^2$, $\Bbf, \Rbf \in \cM^0$. The goal is to find a distributed controller $\Kbf \in \cM^r$, which is stabilizing in the sense that it ensures the system's state converges to a stable equilibrium, thereby minimizing the infinite-horizon average cost among all agents, 
\mequa
\min_{\Kbf} C(\Kbf) & = \min_{\Kbf} \lim_{T \rightarrow \infty} \frac{1}{T}
\sum_{t=0}^{T-1} \mathbb{E}c(t) \\
& = \min_{\Kbf} \lim_{T \rightarrow \infty} \frac{1}{T}
\sum_{t=0}^{T-1} \mathbb{E}\frac{1}{n}\sum_{i=1}^n c_i(t) \\
& =: \min_{\Kbf} \frac{1}{n}\sum_{i=1}^n C_i(\Kbf)
\mequa
where we define $C_i(\Kbf)= \lim_{T \rightarrow \infty} \frac{1}{T} \sum_{t=0}^{T-1} \EE c_i(t)$. 

To illustrate the localized structure of the policy optimization problem for the networked LQR networked system, we present two schematic diagrams (Fig.~\ref{fig:figure_diagram1} and Fig.~\ref{fig:figure_diagram2}), where the relationship between the local perspective and the global perspective is shown clearly.

\subsection{Critical Concepts in LQR and RL}
\label{subsec:keyconceptsin}
To enhance the comprehension of the forthcoming derivation, this subsection introduces some important concepts related to RL and LQR. First, we rewrite the state transition dynamics from the global perspective, 
\#\label{equ: markov chain}
\xbf(t+1)=(\Abf-\Bbf \Kbf)\xbf(t)+\mathbf{\epsilon}(t), 
\#
where $\mathbf{\epsilon}(t) \sim N(\mathbf{0}, \bPsif)$
and $\bPsif = \bPhif+\sigma_0^2 \cdot \Bbf \Bbf^\Tb$, following the problem formulation in Section~\ref{subsection:problemformulation}.

In the setting of LQR, for any $\Kbf \in \RR^{d_u \times d_x}$ such that $\rho(\Abf-\Bbf \Kbf)<1$, let $\Pbf_\Kbf$ be the unique positive definite solution to the Lyapunov equation
\#\label{equ:definition of P_K in Bellman equation}
\Pbf_\Kbf = (\Qbf+\Kbf^\Tb \Rbf \Kbf) + (\Abf-\Bbf \Kbf)^\Tb \Pbf_\Kbf (\Abf-\Bbf \Kbf).
\#

Note that when $\rho(\Abf-\Bbf \Kbf)<1$, the Markov chain in (\ref{equ: markov chain}) has stationary distribution $\nu_\Kbf \sim N(\mathbf{0}, \mathbf{\Xi}_
\Kbf)$, where $\mathbf{\Xi}_\Kbf$ is the unique positive definite solution to the Lyapunov equation
\#\label{sigma_K}
\mathbf{\Xi}_\Kbf = \bPsif + (\Abf-\Bbf \Kbf) \mathbf{\Xi}_\Kbf (\Abf-\Bbf \Kbf)^\Tb.
\#

With a fixed controller $\Kbf \in \cM^r$, the relative action-value function is defined as
\begin{equation}\label{equ: definition of Q_K globally 1}
\begin{aligned}
Q_\Kbf(\xbf,\ubf) = \sum_{t=0}^{\infty} \, \EE\,\big[c(t) &- C(\Kbf)\,\big|\, \xbf(0) = \xbf, \ubf(0) = \ubf, \\ 
&  \ubf(t)\sim  \pi_\Kbf(\cdot|\xbf(t)), t \geq 1\big] .
\end{aligned}
\end{equation}
In the setting of LQR, the state-value function is quadratic for a fixed controller $\Kbf$ (Proposition 3.1, \cite{yang2019provably}),
\mequa
 Q_\Kbf(\xbf,\ubf) & =  \xbf^\Tb (\Qbf+\Abf^\Tb \Pbf_\Kbf \Abf) \xbf + \xbf^\Tb (\Abf^\Tb \Pbf_\Kbf \Bbf) \ubf \\
& \quad + \ubf^\Tb (\Bbf^\Tb \Pbf_\Kbf \Abf) \xbf
 + \ubf^\Tb (\Rbf+\Bbf^\Tb \Pbf_\Kbf \Bbf)\ubf \\
 & \quad - \tr(\Pbf_\Kbf \mathbf{\Xi}_\Kbf) - {\sigma_0}^2\tr(\Rbf+\Pbf_\Kbf \Bbf \Bbf^\Tb) .
\mequa
We further define the \emph{localized}  $Q$ function $Q_{\Kbf}^i$: 
\mequa
Q_\Kbf^i(\xbf, \ubf)= \sum_{t=0}^{\infty}  \,\EE\,\big[c_i(t)-C_i(\Kbf)\,\big|&\, \xbf(0) = \xbf,  \ubf(0) = \ubf, \\
& \ubf(t)\sim  \pi_\Kbf(\cdot|\xbf(t)), t \geq 1\big] .
\mequa
It is straightforward that $Q_{\Kbf}(\xbf,\ubf)=\frac{1}{n} \sum_{i=1}^{n} Q_{\Kbf}^i(\xbf,\ubf)$. Now we can introduce the policy gradient theorem in LQR.  

\begin{lemma} \label{lem:policy gradient theorem} \cite{yang2019provably}
[Policy Gradient Theorem]
\mequa
\nabla_{\Kbf} C(\Kbf) = \EE_{\xbf \sim \nu_\Kbf, \ubf\sim \pi_\Kbf(\cdot | \xbf)} \big[\nabla_{\Kbf} \log \pi_\Kbf(\ubf|\xbf)\cdot Q_\Kbf(\xbf,\ubf)\big] .
\mequa
In LQR, it can be equivalently expressed as
\mequa
\nabla_{\Kbf} C(\Kbf) = 2[(\Rbf+\Bbf^\Tb \Pbf_\Kbf \Bbf)\Kbf -\Bbf^\Tb \Pbf_\Kbf \Abf] \mathbf{\Xi}_\Kbf .
\mequa
\end{lemma}

\section{Main Results}
\label{sec:mainresults}
In this section, we first illustrate the localized gradient approximation method by revisiting the Exponential Decay Property~\cite{qu2019exploiting,qu2020average,qu2020scalable} in the networked LQR setting. Then, we present an algorithmic  framework to solve the networked LQR control problem by using distributed policy gradient descent (Algorithm~\ref{alg:main}). We establish  the theoretical bound of the performance gap between the controller $\Kbf(T)$ obtained by $T$-step gradient descent and the centralized optimal controller $\Kbf^*$, given the localized approximation of the gradient derived from the Exponential Decay Property. Additionally, since the gradient approximation introduces approximation errors in every descent step and the stability may fail, we prove that all the controllers generated by the gradient descent algorithm are stabilizing if the communication range $\kappa$ is large enough and step size $\eta$ is chosen appropriately. Finally, we present the conditions for Exponential Decay Property to hold and provide several representative examples to verify their feasibility.

\subsection{Localized Gradient Approximation}
\label{subsection:exponentialdecayproperty}
This subsection presents the localized gradient approximation method for the development of the scalable algorithm. We revisit a pivotal property, the Exponential Decay Property~\cite{qu2019exploiting,qu2020average,qu2020scalable}, within the LQR framework, which is essential for realizing the localized gradient approximation.

First, we show why we can not access the accurate gradient directly from a local perspective. According to Lemma~\ref{lem:policy gradient theorem}, to get the gradient for $\Kbf_i$, we have,
\mequa
\nabla_{\Kbf_i} C(\Kbf) =   \EE_{\xbf, \ubf} \big[Q_\Kbf(\xbf,\ubf) \nabla_{\Kbf_i} \log \pi_{\Kbf_i}(\ubf_i|\xbf_{\cN_i^r})\big],
\mequa
where we use the fact that $\pi_{\Kbf}(\ubf|\xbf)=\prod_{i=1}^{N} \pi_{\Kbf_i}(\ubf_i|\xbf_{\cN_i^r})$. $Q_{\Kbf}(\xbf,\ubf)$ depends on the global state $\xbf$ and the global action $\ubf$, but each agent $i$ only has access to the information confined within their 
$\kappa$-hop neighborhood ($\cN_i^{\kappa}$). Consequently, calculating the precise gradient from a local standpoint becomes challenging due to the the communication range limitations.

\begin{remark}
We set the control range $r$ and the communication range $\kappa$ as two different variables to separately assess their impacts on the performance degradation respectively. The analysis setting also follows some previous works such as~\cite{alfano2021dimension, shin2022near}. Note that in real-world scenarios, the communication range $\kappa$ and the control range $r$ are usually the same.
\end{remark}

However, the sparsity brought by the network structure makes it possible to employ localized gradient approximation. It is found that the effect from agents far away diminishes with the growth of the distance, therefore, when estimating the Q-function and computing the gradient, the influence of distant agents might be ignored at an acceptable cost. Such a property is called Exponential Decay Property and has been studied in the tabular RL case~\cite{qu2019exploiting, qu2020average,qu2020scalable}. In the following, we re-clarify such a property in the context of the LQR setting. 

\begin{definition}[Exponential Decay Property]
\label{def: Exponential Decay Property}
For a fixed controller $\Kbf$, the Exponential Decay Property holds if,  $\forall i\in \cN$, $\xbf=(\xbf_{\cN_i^\kappa},\xbf_{\cN_{-i}^{\kappa}})$,
$\xbf'=(\xbf_{\cN_i^\kappa},\xbf'_{\cN_{-i}^{\kappa}})$, $\ubf=(\ubf_{\cN_i^\kappa},\ubf_{\cN_{-i}^{\kappa}})$, 
$\ubf'=(\ubf_{\cN_i^\kappa},\ubf'_{\cN_{-i}^{\kappa}})$,
the localized $Q_{\Kbf}^i$ satisfies,
\mequa
\big|Q_{\Kbf}^i(\xbf,\ubf)-Q_{\Kbf}^i(\xbf',\ubf')\big|  \leq  \CC(\xbf,\ubf) \rho^{\kappa+1},
\mequa
for some $\CC(\xbf,\ubf)>0$ that depends on $\xbf$ and $\ubf$, and some $\rho\in(0,1)$.
\end{definition}

This property illustrates that with increasing distance between two agents, the mutual dependence on $Q_{\Kbf}^i$ diminishes exponentially. Consequently, we explore a class of \emph{truncated} $Q$ functions, wherein the dependence on distant agents is selectively truncated,
\begin{equation}
\label{equ:truncatedrelativefunction}
\begin{aligned}
\QQ^i_{\Kbf, \kappa}(\xbf_{\cN_i^\kappa}, \ubf_{\cN_i^\kappa}) = & \int \int \zeta_i(\xbf_{\cN_{-i}^{\kappa}}, \ubf_{\cN_{-i}^\kappa}; \xbf_{\cN_{i}^\kappa}, \ubf_{\cN_{i}^\kappa}) \\
& \hspace{-1cm} Q_{\Kbf}^i(\xbf_{\cN_{i}^\kappa}, \xbf_{\cN_{-i}^{\kappa}}, \ubf_{\cN_{i}^\kappa}, \ubf_{\cN_{-i}^\kappa}) d \xbf_{\cN_{-i}^{\kappa}} d \ubf_{\cN_{-i}^\kappa}.
\end{aligned}
\end{equation}
$\zeta_i(\xbf_{\cN_{-i}^{\kappa}}, \ubf_{\cN_{-i}^\kappa}; \xbf_{\cN_{i}^\kappa}, \ubf_{\cN_{i}^\kappa})$ are distributions conditioned on $\xbf_{\cN_{i}^\kappa}, \ubf_{\cN_{i}^\kappa}$, $\forall \xbf_{\cN_{i}^\kappa}, \ubf_{\cN_{i}^\kappa}$, satisfying,
\mequa
& \int \int \zeta_i(\xbf_{\cN_{-i}^{\kappa}}, \ubf_{\cN_{-i}^\kappa}; \xbf_{\cN_{i}^\kappa}, \ubf_{\cN_{i}^\kappa})d \xbf_{\cN_{-i}^{\kappa}} d \ubf_{\cN_{-i}^\kappa}=1 .
\mequa

In the following lemma, we will show that $\QQ^i_{\Kbf, \kappa}(\xbf_{\cN_i^\kappa}, \ubf_{\cN_i^\kappa})$ is a good approximation of $Q^i_\Kbf(\xbf,\ubf)$. It also inspires us to obtain a good approximation of the exact policy gradient in Lemma~\ref{lem:policy gradient theorem}, as shown in the following lemma.
\begin{lemma}\label{lem:gradient_approximation}
Under the Exponential Decay Property as defined in Definition~\ref{def: Exponential Decay Property}, for any truncated $Q$ function in the form of (\ref{equ:truncatedrelativefunction}), the following holds

\quad (a) For any $(\xbf,\ubf) \in \cX \times \cU$,
\mequa
\left| Q_{\Kbf}^i(\xbf,\ubf)-\QQ^i_{\Kbf, \kappa}(\xbf_{\cN_i^\kappa}, \ubf_{\cN_i^\kappa}) \right| \leq \CC(\xbf,\ubf) \rho^{\kappa+1} .
\mequa

\quad (b) Define the following approximated policy gradient,
\mequa
\hat{\hbf}_i(\Kbf)= 
 \frac{1}{n}\EE_{\xbf\sim \nu_\Kbf, \ubf\sim  \pi_{\Kbf}(\cdot | \xbf)}\sum_{j\in \cN_i^\kappa} &\left[ \QQ^j_ {\Kbf, \kappa}(\xbf_{\cN_j^\kappa}, \ubf_{\cN_j^\kappa}) \right.\\
&  \left. \nabla_{\Kbf_i} \log \pi_{\Kbf_i}(\ubf_i|\xbf_{\cN_i^r}) \right] .
\mequa
where $\nu_\Kbf$ is the stationary distribution for the Markov chain in (\ref{equ: markov chain}). Assuming that $\EE_{\xbf\sim \nu_\Kbf, \ubf\sim  \pi_{\Kbf}(\cdot | \xbf)} \Vert \nabla_{\Kbf_i} \log \pi_{\Kbf_i}(\ubf_i|\xbf_{\cN_i^r}) \Vert \leq L_i$, it follows that $\Vert \hat{\hbf}_i(\Kbf)-\nabla_{\Kbf_i} C (\Kbf) \Vert \leq \CC L_i \rho^{\kappa+1}$, with $\CC$ being a positive constant.

\end{lemma}
\begin{proof}
Please refer the detail in Part A, Section \uppercase\expandafter{\romannumeral3} in~\cite{supplementary}. 
\end{proof}
However, as the LQR belongs to the family of average reward MDP, Exponential Decay Property does not hold universally~\cite{qu2020average}, as opposed to the tabular cases with discounted rewards~\cite{qu2020scalable}.  We will present the mild conditions for such a property to hold in Section~\ref{subsec:conditionsforexponential} and give some representative examples to show their feasibility. $\CC(\xbf,\ubf)$ is a function related to the system parameter $\Abf$, $\Bbf$, $\Qbf$, $\Rbf$, $\Vert \xbf \Vert_{\infty}$ and $\Vert \ubf \Vert_{\infty}$ and the explicit form is given in Section~\uppercase\expandafter{\romannumeral4} in~\cite{supplementary}. Such a powerful property allows us to approximate the exact global gradient locally and accurately in a large network system. 

\subsection{Distributed Policy Gradient}
\label{subsec:optimalguarantee}

\begin{algorithm}[tb]
   \caption{Distributed policy gradient descent}
   \label{alg:main}
\begin{algorithmic}
   \STATE {\bfseries Parameters:} step size $\eta$, communication range limit $\kappa$.
   \STATE {\bfseries Initialization:} Initial controller $\Kbf_1(0)$, $\Kbf_2(0)$, ..., $\Kbf_n(0)$.
   \FOR{step $t=1$ {\bfseries to} $T$}
   \FOR{agent $i=1$ {\bfseries to} $n$}
   \STATE \emph{policy evaluation}: calculate $\QQ^i_{\Kbf, \kappa}(\xbf_{\cN_i^\kappa}, \ubf_{\cN_i^\kappa})$ using the communication within $\kappa$-hop neighbors,
   \STATE \emph{gradient approximation}\footnotemark: calculate $\hat{\hbf}_i(t-1)$,
   \STATE \emph{policy improvement}: $$\Kbf_i(t)=\Kbf_i(t-1)-\eta \hat{\hbf}_i(t-1).$$ 
   \ENDFOR
   \ENDFOR
\end{algorithmic}
\end{algorithm}
\footnotetext{In this study, we employ a simple Monte-Carlo method for policy evaluation and gradient approximation, as detailed in Section~\uppercase\expandafter{\romannumeral8} in~\cite{supplementary}. This framework allows for the integration of alternative approaches, such as actor-critic methods. A comprehensive analysis of these methodologies is part of our planned future work.}

In Algorithm~\ref{alg:main}, we demonstrate a distributed policy gradient descent method. The next theorem shows that it can converge to a global near-optimal point, assuming that each agent has access to a local approximation $\hat{\hbf}_i(\Kbf)$ for the exact policy gradient $\nabla_{\Kbf_i} C(\Kbf)$.  For simplicity, we use $d$ to denote $\max(d_x, d_u)$ and $\mu$ to denote $\sigma_1(\bPsif)$.

\begin{figure*}
\begin{equation}
\label{condition:stepsizecondition}
\begin{aligned}
\eta < \min & \left\{  \underbrace{\frac{1}{16}(\frac{\sigma_1(\Qbf)\mu}{C(\Kbf)})^2\frac{1}{\Vert \Bbf \Vert \Vert \cP_{\cM^r}(\nabla_{\Kbf} C(\Kbf)) \Vert (1+\Vert \Abf - \Bbf \Kbf \Vert)} }_{\textit{To bound } \Vert \mathbf{\Xi}_{\Kbf'}-\mathbf{\Xi}_{\Kbf} \Vert},\,
\underbrace{\frac{\sigma_1(\Qbf)}{32C(\Kbf)\Vert \Rbf+\Bbf^\Tb \Pbf_\Kbf \Bbf \Vert}}_{\textit{To guarantee the convergence}},\right.\\
&\left.
\underbrace{\frac{ \sigma_1(\Qbf) \mu}{4 C(\Kbf) \Vert \Bbf \Vert (\Upsilon(C(\Kbf))+1) \CC \sqrt{d} \sum_i L_i \rho^{\kappa+1}}}_{\textit{To guarantee that } \Kbf'' \textit{is stabilizing} }, \underbrace{\frac{1}{\LL}}_{\textit{To guarantee that } \Kbf' \textit{is stabilizing}}, \,
\underbrace{\frac{-\varpi_1-\sqrt{{\varpi_1}^2-4 \varpi_1 \varpi_2}}{2\varpi_2}}_{C(\Kbf'')\leq C(\Kbf)}, \,1\right \}
\end{aligned}
\end{equation}

\vspace*{12pt}

\begin{equation}
\label{equ:bar L}
\begin{aligned}
\LL \coloneqq \left( 2\sigma_{n}(\Rbf) +  \frac{2\Vert \Bbf \Vert^2 C(\Kbf(0))}{\sigma_1(\bPsif)} + 4\sqrt{2} \zeta \Vert \Bbf \Vert \frac{C(\Kbf(0))}{\mu} \right) \frac{C(\Kbf(0))}{\sigma_{n}(\Qbf)}
\end{aligned}
\end{equation}
\hrulefill
\end{figure*}

\begin{theorem}
\label{theorem:main theorem}
Assuming that $\Kbf^*$ is the centralized optimal controller, if all agents conduct the policy update in Algorithm~\ref{alg:main}, and Exponential Decay Property holds during the process, then for an appropriate choice of the step-size $\eta$ that satisfies (\ref{condition:stepsizecondition}) and an adequate communication range $\kappa$ that,   
\begin{equation}\label{condition:kappa low limit}
\begin{aligned}
\kappa >\frac{1}{-\log{\rho}} \log{\frac{2\sqrt{d} C(\Kbf) \cE(C(\Kbf)) \CC \sum_i L_i}{\Vert \cP_{\cM^r}(\nabla_{\Kbf} C(\Kbf)) \Vert^2 \sigma_{1}(\Qbf)}}-1,
\end{aligned}
\end{equation}
given any arbitrarily small positive real number $\epsilon$, if we conduct the process for $T$ steps that,
\# \label{condition:T condition}
T \geq \frac{\Vert \mathbf{\Xi}_{\Kbf^*} \Vert}{\eta \mu^2 \sigma_{1}(\Rbf)} \log \frac{C(\Kbf(0))-C(\Kbf^*)}{\epsilon} ,
\#
the distributed gradient descent enjoys the following performance bound \footnote{The explicit forms of certain terms unmentioned so far ($\cE(\cdot)$, $\Upsilon(\cdot)$, $M_1 F_1(\cdot)$, $M_2 F_2(\cdot)$) are detailed in Section~\uppercase\expandafter{\romannumeral1} in~\cite{supplementary}.}
\begin{equation}\label{result:performance of C(K(T))}
\begin{aligned}
& C(\Kbf(T))-C(\Kbf^*) \leq  \epsilon + \\
& \frac{\Vert \mathbf{\Xi}_{\Kbf^*} \Vert}{\eta \mu^2 \sigma_{1}(\Rbf)} \big[M_1 F_1(C(\Kbf(0))) \rho^{\kappa+1} + M_2 F_2(C(\Kbf(0)))\rho^r\big].
\end{aligned}
\end{equation} 
\end{theorem}

In the following text, we demonstrate the proof process step by step. \footnote{Due to the limitation of the length of the main text, some lemmas and corollary from previous works will be directly cited and used. The complete and coherent proof can be found in Section~\uppercase\expandafter{\romannumeral2} in~\cite{supplementary}.}
The one-step exact gradient descent (gradient without approximation error) is,
\begin{equation}\label{equ:def of K'}
\Kbf'=\Kbf-\eta \cP_{\cM^r}(\nabla_{\Kbf} C(\Kbf)),
\end{equation}
where $\cP_{\cM^r}$ denotes the projection on the sub-space $\cM^r$ (Definition~\ref{def:Mr}). The one-step distributed gradient descent in Algorithm~\ref{alg:main} is,
\begin{equation}\label{equ:def of K''}
\Kbf_i''=\Kbf_i - \eta \hat{\hbf}_i(\Kbf) .
\end{equation}
We use $\Kbf''\in \RR^{d_u \times d_x} \cap \cM^r$ to denote the controller matrix consisting of all $\Kbf_i''\in \RR^{d_u^i \times d_x^i}$ (like Fig.~\ref{fig:figure_diagram1}). Our proof mainly consists of four parts. 
\begin{enumerate}
    \item[1)] To prove that with a stabilizing $\Kbf$ and an appropriate choice of step size $\eta$, $\Kbf'$ is stabilizing.
    \item[2)] To prove that with a stabilizing $\Kbf'$ and an appropriate choice of step size $\eta$, $\Kbf''$ is stabilizing.
    \item[3)] To prove that with a stabilizing $\Kbf$, an appropriate choice of step size $\eta$ and an adequate communication range $\kappa$, $C(\Kbf'')\leq C(\Kbf)$.
    \item[4)] To prove that the generated controller $\Kbf(T)$ converges to the point close to the centralized optimal controller $\Kbf^*$ but suffers from the degradation exponentially small in $\kappa$ and $r$. 
\end{enumerate}

Each part leads to one or several restrictions on the choice of the step size $\eta$, which is shown in (\ref{condition:stepsizecondition}). We will explain the meaning and the insight for each term in the following text.  It's important to acknowledge that while the step-size condition outlined in (\ref{condition:stepsizecondition}) appears to be contingent upon $\Kbf$ and varies over time with each update iteration, it is feasible to establish a universal lower bound through the principle of monotonicity. We leave the discussion to Section~\uppercase\expandafter{\romannumeral2}.C to E in the supplementary material~\cite{supplementary}.

\subsection{Stability and Descent Guarantee}
\label{subsec:Stabilizability and Descent Guarantee}

This subsection presents the generated controllers' stability and the descent guarantee of the objective function (Step 1-3). 

First, we illustrate that with an  appropriate  choice of $\eta$ and starting from a stabilizing controller $\Kbf$, $\Kbf'$ obtained by the one-step exact policy gradient descent (\ref{equ:def of K'}) is stabilizing as well. 

\begin{corollary} 
\label{corollary:K' Stabilizing}
If $\Kbf$ is stabilizing, the sub-level set $S_{C(\Kbf)}$, which is defined as $\{ \Kbf': C(\Kbf')<C(\Kbf), \rho(\Abf - \Bbf \Kbf') <1, \Kbf'\in \cM^r \}$,  is compact, and the following conclusion hold:
There exists a constant $\LL$ which is determined by $\Abf$, $\Bbf$, $\Qbf$. $\Rbf$, $\bPsif$ and $C(\Kbf(0))$ that $C(\Kbf)$ is a $\LL$-smooth function in the projection space $\cM^r$, and if we choose a step size $\eta \leq 1/\LL$, it is guaranteed that $\Kbf' = \Kbf-\eta \cP_{\cM^r}(\nabla_{\Kbf} C(\Kbf)) $ stays in the stabilizing set $S_{C(\Kbf)}$. (\cite{bu2019lqr}, Lemma 7.4 $\&$ 7.9)
\end{corollary}

This corollary guarantees the stability of $\Kbf'$ given a stabilizing $\Kbf$. Note that the non-convexity of the problem introduces additional complications for determining the explicit form of $\LL$. However, it has been shown in a previous work~\cite{bu2019lqr} that an internal property (coerciveness) of $C(\Kbf)$ remedies the complication. The explicit form of $\LL$ is provided in (\ref{equ:bar L}). Guarantee on the stability of $\Kbf'$ leads to the fourth term in (\ref{condition:stepsizecondition}).

Based on the stabilizing $\Kbf'$, we then show the stability of $\Kbf''$. Conclusively speaking, if $\Vert \Kbf'' - \Kbf' \Vert$ is small enough, $\Kbf''$ is also stabilizing. By the localized gradient approximation, we have,
\# \label{equ:difference between K'' and K'}
\Vert \Kbf''-\Kbf' \Vert = \eta \big\Vert \sum_{i=1}^n \nabla_{\Kbf_i} C(\Kbf) - \hat{\hbf}_i(\Kbf) \big\Vert ,
\#
where Exponential Decay Property can be utilized to measure the gap between $\nabla_{\Kbf_i} C(\Kbf)$ and $\hat{\hbf}_i(\Kbf)$. Therefore, the difference $\Vert \Kbf''-\Kbf' \Vert$ is determined by two factors: step size $\eta$ and the communication range $\kappa$, which will be clarified in the following corollary. 
\begin{corollary}
\label{corollary: Kbf'' stabilizing}
If it is ensured that 
$$\eta < \frac{ \sigma_1(\Qbf) \mu}{4 C(\Kbf) \Vert \Bbf \Vert (\Upsilon(C(\Kbf))+1) \CC \sqrt{d} \sum_i L_i \rho^{\kappa+1}},$$ 
with a stabilizing $\Kbf'$ and (\ref{equ:difference between K'' and K'}), we can claim that $\Kbf''$ is stabilizing as well. 
\end{corollary}
\begin{proof}
See Appendix~\ref{appen:proof of corollary Kbf'' stabilizing}.
\end{proof}

Knowing that $\Kbf'$ and $\Kbf''$ are both stabilizing, here we further provide the descent guarantee of the objective function $C(\Kbf)$.
Note that the stability of $\Kbf''$ is not equivalent to the fact that the controller's performance is guaranteed to be improved by each iteration: If $\Kbf''$ is stabilizing, $C(\Kbf'')$ is finite, but $C(\Kbf'')\leq C(\Kbf)$ is not guaranteed. We illustrate that the communication range limit $\kappa$ needs to meet a theoretically lower limit to ensure performance improvement. To see this, we keep on decomposing the one-step change of objective function into two parts,
\mequa
C(\Kbf'')-C(\Kbf)=[C(\Kbf'')-C(\Kbf')]+[C(\Kbf')-C(\Kbf)].
\mequa
The first term represents the approximation error introduced by the local approximated gradient, and the second term represents the descent brought up by the exact gradient. Next, we analyze them separately. 

\begin{corollary}
\label{corollary:diff between C(K') and C(K)}
For $C(\Kbf')-C(\Kbf)$, following the descent lemma for an $\LL$-smooth function~\cite{beck2017first}, if the step size $\eta$ is smaller than $2/\LL$, the controller's performance by the one-step exact gradient descent enjoys the following improvement: 
\mequa
C(\Kbf')-C(\Kbf) \leq -(\eta-\frac{\LL}{2}\eta^2) \big\Vert \cP_{\cM^r}(\nabla_{\Kbf} C(\Kbf)) \big\Vert^2.
\mequa
\end{corollary}
\begin{proof}
The gradient descent update gives $\Kbf'=\Kbf-\eta \cP_{\cM^r} \nabla C(\Kbf)$, and we have,
\mequa
C(\Kbf')\leq & C(\Kbf)+\langle \cP_{\cM^r} \nabla C(\Kbf), \Kbf'-\Kbf \rangle + \frac{\LL}{2} \Vert \Kbf'-\Kbf \Vert^2 \\
= & C(\Kbf)+\langle \cP_{\cM^r} \nabla C(\Kbf), -\eta \cP_{\cM^r} \nabla C(\Kbf) \rangle \\
& + \frac{\LL}{2} \Vert \eta \cP_{\cM^r} \nabla C(\Kbf) \Vert^2 \\
= & C(\Kbf)-(\eta - \frac{\LL}{2}\eta^2) \Vert \cP_{\cM^r} \nabla C(\Kbf) \Vert^2 \leq C(\Kbf).
\mequa
\end{proof}

For $C(\Kbf'')-C(\Kbf')$, we prove that the $C(\Kbf'')-C(\Kbf')$ can be bounded by a polynomial linear combination of $\eta \rho^{\kappa+1}$, $(\eta \rho^{\kappa+1})^2$ and $(\eta \rho^{\kappa+1})^3$. 
\begin{corollary}
\label{corollary:diff between C(K'') and C(K')}
We can bound the difference between $C(\Kbf'')$ and $C(\Kbf')$ by a polynomial:
\begin{equation} \label{appen:equ: C(K'')-C(K')}
\begin{aligned} 
C(\Kbf'')-C(\Kbf') \leq & \eta f_1^1(C(\Kbf)) \rho^{\kappa+1} + \eta^2 f_1^2(C(\Kbf)) \rho^{2(\kappa+1)} \\
& + \eta^3 f_1^3(C(\Kbf)) \rho^{3(\kappa+1)} ,
\end{aligned}
\end{equation}
the explict forms of $f_1^1(C(\Kbf))$, $f_1^2(C(\Kbf))$ and $f_1^3(C(\Kbf))$ can be found in Section~\uppercase\expandafter{\romannumeral1}.A in~\cite{supplementary}.
\end{corollary}

\begin{proof}
See Appendix~\ref{appen: proof diff between C(K'') and C(K')}.
\end{proof}

By combining the results in the analysis of $C(\Kbf'')-C(\Kbf')$ (Corollary~\ref{corollary:diff between C(K'') and C(K')}) and $C(\Kbf')-C(\Kbf)$ (Corollary~\ref{corollary:diff between C(K') and C(K)}), we obtain an upper bound for $C(\Kbf'')-C(\Kbf)$ in the form of a cubic polynomial concerning $\eta$. It can be represented as:
\# \label{equ:error between C(K) and C(K'')}
C(\Kbf'')-C(\Kbf) \leq -\eta (\varpi_2 \eta^2 +\varpi_1 \eta +\varpi_0) ,
\#
where 
\mequa
& \varpi_2 =- f_1^3(C(\Kbf)) \rho^{3(\kappa+1)}, \\
& \varpi_1 =- f_1^2(C(\Kbf)) \rho^{2(\kappa+1)} + \frac{\LL}{2}\Vert \cP_{\cM^r}(\nabla_{\Kbf} C(\Kbf)) \Vert, \\
& \varpi_0 =\Vert \cP_{\cM^r}(\nabla_{\Kbf} C(\Kbf)) \Vert - f_1^1(C(\Kbf)) \rho^{\kappa+1}.
\mequa
The following corollary gives the conditions to guarantee that $C(\Kbf'')-C(\Kbf)<0$.
\begin{corollary}
\label{corollary: large enough kappa}
If we have 
$$ \kappa >\frac{1}{-\log{\rho}} \log{\frac{2\sqrt{d} C(\Kbf) \cE(C(\Kbf)) \CC \sum_i L_i}{\Vert \cP_{\cM^r}(\nabla_{\Kbf} C(\Kbf)) \Vert^2 \sigma_{1}(\Qbf)}}-1, $$
it is guaranteed that with all the $\eta$ in the following range,
\mequa
0 < \eta \leq \frac{-\varpi_1 - \sqrt{\varpi_1^2 -4 \varpi_2 \varpi_0}}{2\varpi_2},
\mequa
the distributed gradient descent makes the objective function $C(\Kbf)$ decrease. 
\end{corollary}
\begin{proof}
Note that in~\ref{equ:error between C(K) and C(K'')}, $\varpi_2$ is negative. To obtain meaningful values of $\eta$ that satisfy $C(\Kbf'') - C(\Kbf) < 0$, two conditions must be met: 1) $\varpi_0 > 0$, and 2) $\eta$ must be smaller than the positive solution of the quadratic equation $\varpi_2 \eta^2 + \varpi_1 \eta + \varpi_0 = 0$. The first condition establishes a lower limit for $\kappa$, as demonstrated in (\ref{condition:kappa low limit}). The second condition imposes an additional upper limit on $\eta$, which is the fifth term in (\ref{condition:stepsizecondition}). For the detailed proof, please refer to Part D, Section~\uppercase\expandafter{\romannumeral2} in~\cite{supplementary}.
\end{proof}

\begin{remark}
We clarify the factors that affect system stability by decomposing the one-step gradient descent into two parts. $C(\Kbf')-C(\Kbf)$ represents the performance improvement caused by the one-step gradient descent along the direction of steepest descent direction $\cP_{\cM^r} \nabla_{\Kbf} C(\Kbf)$ over a subspace $\cM^r$. It has been shown in previous works that with a small enough step size along the projection of the exact gradient direction, the descent of $C(\Kbf)$ can be guaranteed~\cite{bu2019lqr} and the generated $\Kbf'$ stays in $S_{C(\Kbf)}$. Such a guaranteed performance improvement is represented as a dashed red line in Fig.~\ref{fig:illustrate}.

However, $C(\Kbf'')-C(\Kbf')$ represents the performance disturbance brought about by the inaccurate approximation of the exact global gradient, i.e., the difference between $\hat{\hbf}_i(\Kbf)$ and $\nabla_{\Kbf_i} C(\Kbf)$ for each agent $i$.
Specifically, the smaller the communication range $\kappa$ is, the less accurate each agent approximates the exact gradient using the local information, and the more likely the controller becomes unstabilizing during the descent process. Intuitively speaking, the condition on $\kappa$ (\ref{condition:kappa low limit}) shows that we have to make a gradient approximation $\hat{\hbf}_i(\Kbf)$ accurate enough so that $\Kbf''$ does not deviate from $\Kbf'$ too much. Otherwise, the objective function $C(\Kbf)$ is not guaranteed to keep decreasing.

Note that as the descent of the objective function, $\Vert \cP_{\cM^r}(\nabla_{\Kbf} C(\Kbf)) \Vert$ may decrease, and the required low limit for $\kappa$ in Corollary~\ref{corollary: large enough kappa} may increase accordingly. If $\kappa$ is fixed, the final controller $\Kbf(T)$ may suffer from a minimum fixed performance degradation compared with the optimal controller $\Kbf^*$. We discuss the case in Claim 1 and Claim 2 in the supplementary material~\cite{supplementary}. 
\end{remark}

\begin{figure}[t]
\begin{center}
\centerline{\includegraphics[width=0.9\columnwidth]{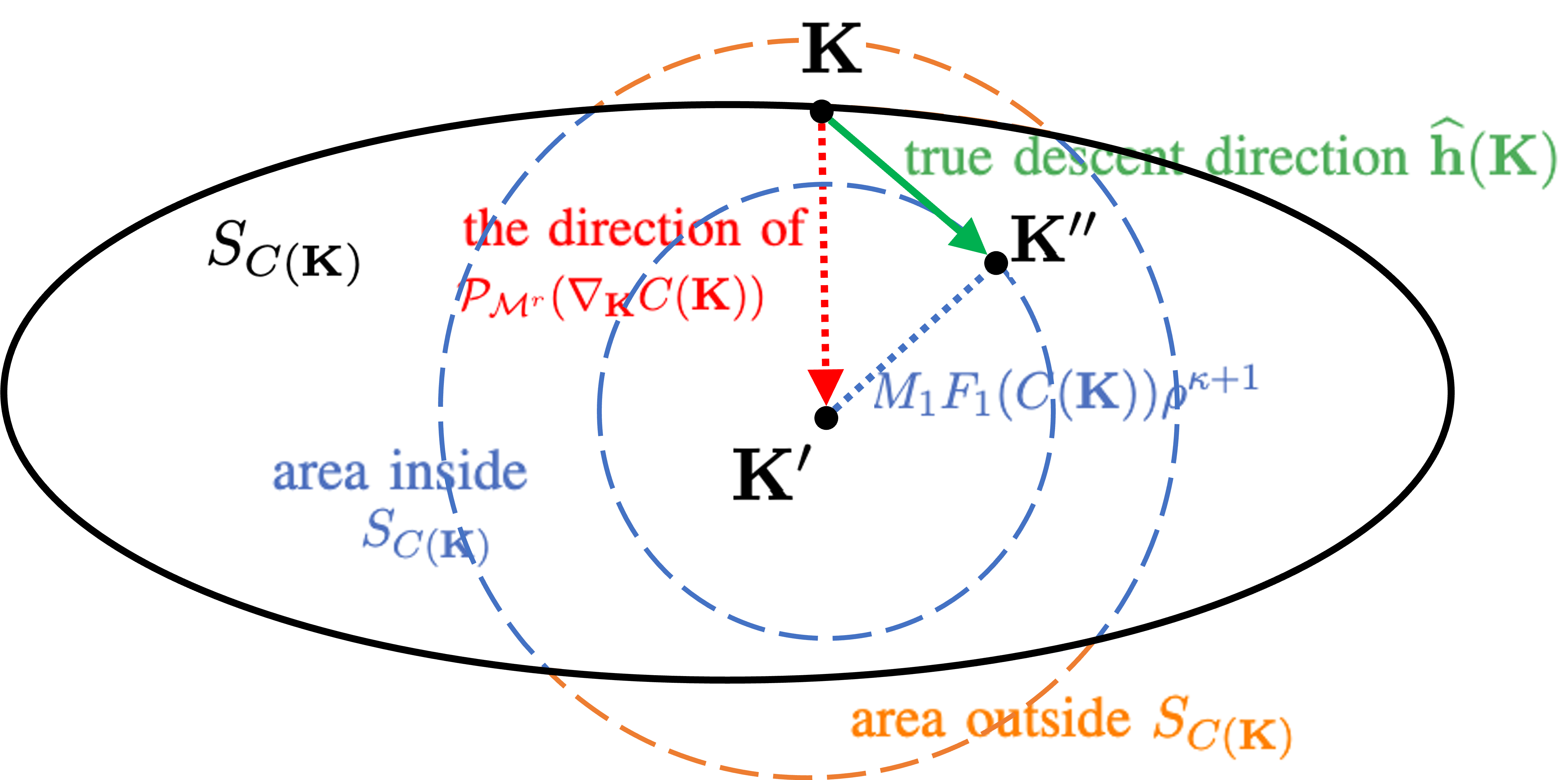}}
\vskip -0.1in
\caption{\small 
The black oval represents the sub-level set $S_{C(\Kbf)}$. The red line and the green line represent the one-step move along the direction of $\cP_{\cM^r}(\nabla_{\Kbf} C(\Kbf))$ and $\hat{\hbf}(\Kbf) = \sum_{i=1}^n \hat{\hbf}_i(\Kbf)$, respectively. The blue line represents the difference caused by the gradient approximation and thus depends on $\kappa$. If the blue circle is so large that $C(\Kbf'')$ moves out of $S_{C(\Kbf)}$ to the orange area, i.e., the $\kappa$ is so small such that the approximation is too inaccurate, the system may take the risk of being unstable.
} 
\label{fig:illustrate}
\end{center}
\vskip -0.3in
\end{figure}

\subsection{Convergence and Degradation Quantification}
\label{subsubsec:convergenceanddegradation}
In this subsection, we quantify the performance gap between the ultimate controller $\Kbf(T)\in \cM^r$ obtained by Algorithm~\ref{alg:main} and the centralized optimal controller $\Kbf^*$ (step 4). 

We keep on using the technique of decomposing $C(\Kbf'')-C(\Kbf)$ into $C(\Kbf'')-C(\Kbf')$ and $C(\Kbf')-C(\Kbf)$. As shown in Corollary~\ref{corollary:diff between C(K'') and C(K')}, we use the Exponential Decay Property (Definition~\ref{def: Exponential Decay Property}) to bound $C(\Kbf'')-C(\Kbf')$. Since $\rho<1$, we can use a single term $M_1 F_1(C(\Kbf)) \rho^{\kappa+1}$ to represent the bound in (\ref{appen:equ: C(K'')-C(K')}) for simplicity:
\mequa
C(\Kbf'')-C(\Kbf') \leq M_1 F_1(C(\Kbf)) \rho^{\kappa+1}.
\mequa

Unlike the stability analysis of $\Kbf'$ and $\Kbf''$, we analyze the term $C(\Kbf')-C(\Kbf)$ to highlight the loss of optimality caused by the projection of the exact gradient $\nabla_{\Kbf} C(\Kbf)$ on $\cM^r$ (Definition~\ref{def:Mr}). We additionally define that:
\mequa
\Kbf^h = \Kbf-\eta \nabla_{\Kbf} C(\Kbf) ,
\mequa
where $\Kbf^h$ is the controller obtained by one-step gradient descent from $\Kbf$ without projection. Consequently, $\Kbf^h$ may not necessarily be a sparse blocked matrix. Considering $C(\Kbf')-C(\Kbf)$, we keep on separating the term into two parts:
\mequa
C(\Kbf')-C(\Kbf) = [C(\Kbf')-C(\Kbf^h)] + [C(\Kbf^h) - C(\Kbf)].
\mequa

First, we analyze $C(\Kbf^h)-C(\Kbf)$.
In some recent works like~\cite{fazel2018global,bu2019lqr}, it has been proved that the global optima can be reached through the policy gradient if the LQR is centralized. In other words, if the controller moves from $\Kbf$ to $\Kbf^h$ every step, with an appropriate choice of $\eta$, the descent process finally reaches the optima. The one-step analysis is as follows:
\begin{corollary}
If we have
\#
\eta \leq \frac{\sigma_1(\Qbf)}{32C(\Kbf)\Vert \Rbf+\Bbf^\Tb \Pbf_\Kbf \Bbf \Vert},
\#
then we have
\#
C(\Kbf^h)-C(\Kbf) \leq -\eta \frac{\mu^2 \sigma_1(\Rbf)}{\Vert \mathbf{\Xi}_{\Kbf^*} \Vert}  (C(\Kbf)-C(\Kbf^*)).
\#
\end{corollary}
\begin{proof}
The methodology for the proof is akin to that used in Lemma 11 of~\cite{fazel2018global} for a centralized LQR problem. For a comprehensive elaboration, refer to Part E1, Section~\uppercase\expandafter{\romannumeral2} in~\cite{supplementary}. This leads to the second term in (\ref{condition:stepsizecondition}). 
\end{proof}

Then, we analyze $C(\Kbf')-C(\Kbf^h)$.
In this term, the factor $\Vert \cP_{\cM^r} (\nabla_{\Kbf} C(\Kbf))-\nabla_{\Kbf} C(\Kbf)\Vert$ plays an important role. It characterizes the effect of projecting the gradient to a subspace and causes the degradation term concerning $\rho^r$ in the ultimate controller's performance gap (\ref{result:performance of C(K(T))}). Therefore, the structural property of $\nabla_{\Kbf} C(\Kbf)$ is crucial. We show that under the Exponential Decay Property, the matrix $\nabla_{\Kbf} C(\Kbf)$ enjoys a special decaying property called $(C_{\nabla_{\Kbf}}, \rho)$-spatially exponential decaying (SED). 
\begin{definition} 
\label{appen:defin:spatially exponential decaying (SED)}
\cite{zhang2022optimal}
[Spatially Exponential Decaying (SED)]
A $n$-by-$n$ blocked matrix $\Xbf$ is $(c,\gamma)$-spatially exponential decaying (SED) if,
\mequa
\Vert [\Xbf]_{ij}\Vert \leq c \cdot \gamma^{\dist(i,j)}, \quad \forall i,j \in \cN ,
\mequa
where $0<\gamma<1$,  $c>0$.
\end{definition}

The constant $C_{\nabla_{\Kbf}}$ is associated with $n$ and other system parameters like $\Abf$, $\Bbf$, $\Rbf$. Due to its complexity, the explicit formulation is detailed in Part C, Section~\uppercase\expandafter{\romannumeral4} in~\cite{supplementary}. The next corollary bound the difference between $C(\Kbf')$ and $C(\Kbf^h)$.

\begin{corollary}
\begin{equation}
\begin{aligned}
& C(\Kbf')-C(\Kbf^h) \leq M_2 F_2(C(\Kbf)) {\rho}^{r}.
\end{aligned}
\end{equation}
\end{corollary}
\begin{proof}
SED describes the exponentially decaying property of the $(i,j)$-th sub-matrix's norm as $\dist(i,j)$ grows. So the key point of the proof is that we can bound $\Vert \nabla_{\Kbf} C(\Kbf) - \cP_{\cM^r}(\nabla_{\Kbf} C(\Kbf)) \Vert$ by a term related to $\rho^r$, because the distance between the two agents corresponding to it is at least $r$.
\mequa
& \big\Vert \nabla_{\Kbf} C(\Kbf) - \cP_{\cM^r}(\nabla_{\Kbf} C(\Kbf)) \big\Vert \\
& \leq \sum_{i=1}^n \sum_{j\in \cN_{i}^{-r}}  \big\Vert [\nabla_{\Kbf} C(\Kbf)]_{ij} \big\Vert_F \\
& \leq \sum_{i=1}^n \sum_{j\in \cN_{i}^{-r}} C_{\nabla_{\Kbf}} \rho^{\dist(i,j)} \leq n \max_i(|\cN_{i}^{-r}|) C_{\nabla_{\Kbf}} \rho^{r}.
\mequa
Using this term, we can bound $C(\Kbf') - C(\Kbf^h)$ employing a technique similar to that in (\ref{equ: bound on CKbf}). The remainder of the proof is detailed in Part E2, Section~\uppercase\expandafter{\romannumeral2} in~\cite{supplementary}.
\end{proof}

By combining the analysis of $C(\Kbf^h)-C(\Kbf)$ and  $C(\Kbf')-C(\Kbf^h)$, we obtain the one-step performance gap compared with $\Kbf^*$ :
\mequa
& C(\Kbf'')-C(\Kbf) \leq - \eta \frac{\mu^2 \sigma_1(\Rbf)}{\Vert \mathbf{\Xi}_{\Kbf^*} \Vert} (C(\Kbf)-C(\Kbf^*)) \\
& + M_1 F_1(C(\Kbf))  \rho^{\kappa+1} + M_2 F_2(C(\Kbf))) \rho^r,
\mequa
Inductively we have,
\mequa
& C(\Kbf(T))-C(\Kbf^*) \\
\leq & \big(1- \eta \frac{\mu^2 \sigma_1(\Rbf)}{\Vert \mathbf{\Xi}_{\Kbf^*} \Vert}\big)^T (C(\Kbf(0)) - C(\Kbf^*)) \\
& + \frac{\Vert \mathbf{\Xi}_{\Kbf^*} \Vert}{\eta \mu^2 \sigma_1(\Rbf)} [M_1 F_1(C(\Kbf(0))) \rho^{\kappa+1} + M_2 F_2(C(\Kbf(0)))\rho^r],
\mequa
Note that both $F_1(C(\Kbf))$ and $F_2(C(\Kbf))$ monotonically increases with the increase of $C(\Kbf)$, so $F_1(C(\Kbf^*)) \leq F_1(C(\Kbf(T))) \leq F_1(C(\Kbf(0)))$, $F_2(C(\Kbf^*) \leq F_2(\Kbf(T)) \leq F_2(C(\Kbf(0)))$. Provided a minimum iteration number $T$ as (\ref{condition:T condition}), we obtain the final result (\ref{result:performance of C(K(T))}) as shown in Section~\ref{appen:subsec:iteration results} in the Appendix.

\begin{remark}
As shown in Theorem~\ref{theorem:main theorem}, the performance degradation compared with the optimal controller consists of two parts. The term related to $\rho^{\kappa+1}$ represents the degradation introduced by the gradient approximation. Each agent only uses the information within its $\kappa$-hop neighborhood to make a relatively accurate approximation of the exact gradient. Then they conduct the distributed policy gradient descent locally, which causes errors in every iteration. The term related to $\rho^r$ represents the degradation brought by the controller's truncation (or projection). Each agent implement its control input depending on its neighbors within $r$-hop, while in centralized controlling, each agent affects any other agent no matter how far away they are. Due to the sparsity of the system matrices $\Abf$, $\Bbf$, $\Qbf$, $\Rbf$, the degradation decreases to $0$ exponentially with the growth of communication range limit $\kappa$ and control range $r$, so at least we can still obtain a near-optimal controller. 

It is significant to see that the graph diameter is not included in the theoretical performance gap in (\ref{result:performance of C(K(T))}), though if we set $r$ as the graph diameter and let $\kappa$ equal to $r$, the multi-agent problem will degenerate into a single-agent problem, and the second and third error terms should be subtracted. The reason is that within our theoretical framework, each node is presumed to operate within an infinitely expansive network. Therefore the diameter of the graph does not appear in our final conclusion. To incorporate the graph's diameter into the final conclusion, it may be essential to consider the graph's specific structure and the boundary conditions. We consider it to be a topic for our future work.

By setting $r$ and $\kappa$ to the graph's diameter, the multi-agent scenario simplifies to a single-agent case, the second and third error terms should be eliminated~\cite{fazel2018global}. But it is not the case and the graph diameter does not factor into the theoretical performance gap as outlined in (\ref{result:performance of C(K(T))}).  It stems from our theoretical assumption that each node operates within a boundlessly network, thereby rendering the graph's diameter irrelevant to our final deductions. To integrate the graph diameter into our conclusions, a detailed examination of the graph's inherent structure and edge conditions may be required. This represents a potential topic for our future research.
\end{remark}

\subsection{Conditions for Exponential Decay Property}
\label{subsec:conditionsforexponential}
This subsection presents the conditions for the Exponential Decay Property (Definition~\ref{def: Exponential Decay Property}) to hold in a networked LQR setting.

\begin{lemma}
\label{lem:conditionsForLQR1}
At time $t$, for any node $j$, suppose that $i,i'$ are $\kappa$, $\kappa+1$-hop neighbors of $j$, respectively, for any $\Kbf$ generated from the gradient descent process, if there exists $\rho \in (0,1)$ such that the following inequality holds,
\mequa
\left \Vert [(\Abf-\Bbf \Kbf)^t]_{i'j} \right \Vert \leq \rho \left \Vert [(\Abf-\Bbf \Kbf)^t]_{ij} \right \Vert,
\mequa
then the $(C_{Q_{\Kbf}^i},\rho)$ Exponential Decay Property holds.
\end{lemma}

This lemma indicates some specific spatial decay characteristics brought by the structure of the state transition matrix $\Abf - \Bbf \Kbf$. However, determining whether the condition can be satisfied is challenging, even when provided with the graph structure and transition dynamics model. This difficulty arises as it is impractical to iterate over all $t$ for verification. So in the following lemma, a stricter but more intuitive condition is presented, which clarifies the nature of the property in the LQR setting more directly.

\begin{lemma}
\label{lem:conditionsForLQR2}
With a stabilizing controller $\Kbf \in \cM^r$, for any agent $i$ and agent $j$ that are $\kappa$-hop neighbors, if there exist constants $C>0$, $\cD>0$, and $\rho\in(0,1)$ such that 
$|\cW_{i\rightarrow j}^t(r) | \leq C  \cD^{t} {\rho}^{\kappa}$ and $\overline{[\Abf - \Bbf \Kbf]} \cdot \cD \leq 1$ holds, then the $(C_{Q_{\Kbf}^i}, \rho)$-Exponential Decay Property holds. 
\end{lemma}

The detailed proof of these two lemma and the explicit form of $C_{Q_{\Kbf}^i}$ are deferred to Section~\uppercase\expandafter{\romannumeral4} in~\cite{supplementary}.
$\cW_{i\rightarrow j}^t(r)$ is defined as a set that contains all the walks from $i$ to $j$ with length $t$ (Definition~\ref{appen:def:walk number},  Appendix~\ref{appen:subsec:helper definitions and lemmas}). Such a property yields the intuition that the decay property results from the weak interaction between agents and the low connectivity of the underlying network. In particular, on the one hand, $\overline{[\Abf - \Bbf \Kbf]}$ represents the maximum impact between two agents in the transition process. On the other hand, the network should be as sparse as possible so that $\cD$ is small enough. With a fixed $t$, the number of the available walks ($| \cW_{i\rightarrow j}^t |$) is supposed to decrease exponentially with the growth of the distance $\kappa$ between $i$ and $j$. The conditions above are consistent with the assumptions in~\cite{qu2020average} concerning the tabular case and average-reward setting  in RL.

These additional conditions show that the Exponential Decay Property does not hold uniformly in the networked LQR settings. Fortunately, the cardinality of $\cW_{i\rightarrow j}^t(r)$ ($|\cW_{i\rightarrow j}^t(r) |$) has an explicit mathematical expression in certain regular and representative graphs where we can directly calculate $C$, $\cD$ and $\rho$, so there are plenty of cases where we can verify and apply this property. In Appendix~\ref{subsec:Special Cases Study for Exponential Decay Property}, we present the analytic and numerical results in some representative graph structures.

\begin{remark}
The conditions in Lemma~\ref{lem:conditionsForLQR2} are stricter than those in Lemma~\ref{lem:conditionsForLQR1}. In the proof of Lemma~\ref{lem:conditionsForLQR2} (Part B, Section~\uppercase\expandafter{\romannumeral4} in~\cite{supplementary}), the condition $\overline{[\Abf - \Bbf \Kbf]} \cdot \cD \leq 1$ is supplemented to construct a convergent geometric sequence summation, which is not necessary if the system is stable. We leave the finding of a more relaxed  condition as a topic of future work. 
\end{remark}

\section{Simulation Results}

\begin{figure*}[t]
  \centering
  \includegraphics[width=0.8\linewidth]{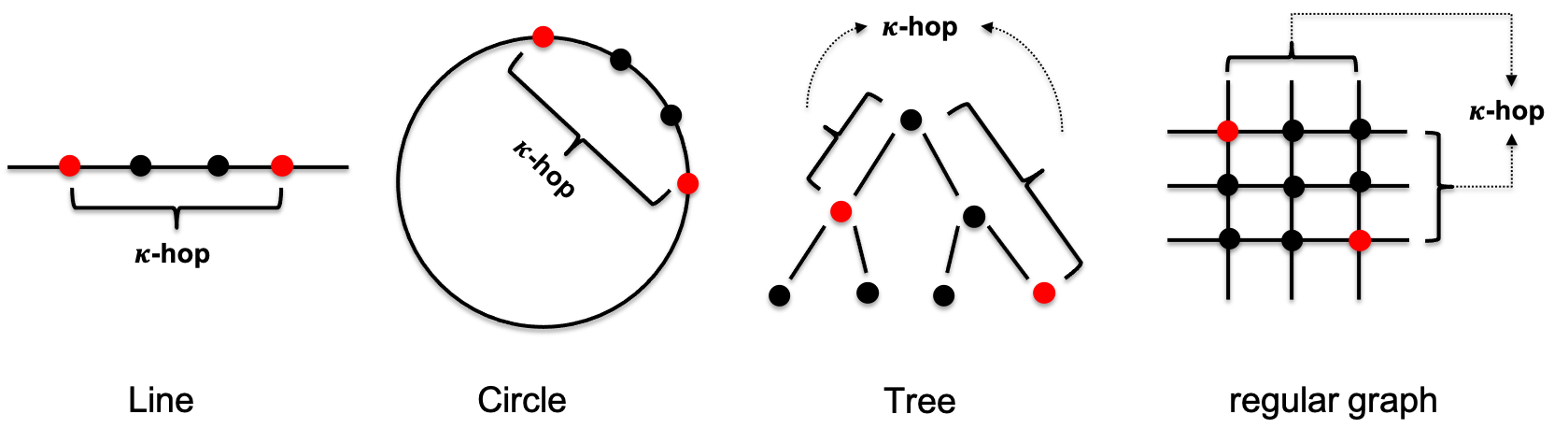}
  \caption{\small The diagram of the four representative graphs: line, circle, 2-ary tree and 4-regular grid.}
  \label{fig:diagram}
\end{figure*}


In this section, we present a set of examples to demonstrate the results reported in this paper. We choose four representative graphs: line, circle, tree, and 4-regular grid. The diagram of these four types of graphs is shown in Fig.~\ref{fig:diagram}. The combination of these graphs can represent a certain quantity of random graphs.

\subsection{Simulation Results for Distributed Policy Gradient Descent}

The system is with parameters $(\Abf, \Bbf)$, where $\Bbf=\Ibf$ and $\Abf$ is initially consistent with the adjacency matrix of each graph and is scaled by 0.9 until $\rho(\Abf)<1$, then the initial controller $\Kbf(0)$ can be set to be $\textbf{0}$. The cost matrices $(\Qbf, \Rbf)$ are set to be identity with appropriate dimensions. We set step size $\eta=0.001$. The noise covariance matrix $\bPsif$ is set to be $0.5 \Ibf$. The total iteration $T$ is set to be 4000. In the line and the circle graph, we set 99 nodes. In the 2-ary tree graph, we set 127 nodes to construct an 8-layer full binary tree. In the 4-regular grid, we set $11\times 11=121$ nodes. We conduct the distributed policy gradient descent as Algorithm~\ref{alg:main}.

We set up different experiments to verify the impact of the communication range $\kappa$ and control range $r$, respectively. In the experiments regarding $r$, we set $\kappa$ to be maximum, which means the gradient approximation is accurate. In the experiments regarding $\kappa$, $r$ is set to be maximum, which means that there will be no gradient truncation. 

As is shown in Fig.~\ref{fig: relative performance gap with kappa and r}, we explore the relative cost error \(\frac{C(\Kbf(T))-C(\Kbf^*)}{C(\Kbf^*)}\), under varying communication range limits, \(\kappa\), and control range, \(r\). On the left, with \(r\) maximized to prevent gradient truncation and a Gaussian-distributed approximation error added to \(\nabla_{\Kbf_i} C(\Kbf)\), we examine \(\kappa \in [2,3,5,10,20]\). On the right, setting \(\kappa\) to its maximum ensures every agent \(i\) has access to the exact gradient, with \(r\) varying within \([2,3,5,10,20]\). For example, at \(r=2\), an agent's dynamics depend only on its 2-hop neighbors, whereas \(r=20\) equates to a centralized LQR scenario, allowing optimal convergence. We also show the performance when $\kappa$ and $r$ are less than the diameter of the graph, the results are shown in Fig.3 in the supplementary material~\cite{supplementary}.

In both semi-logarithmic plots with logarithmic y-axis in Fig.~\ref{fig: relative performance gap with kappa and r}, the curve of relative cost error exhibits strong linear properties, which aligns with the theoretical results shown as (\ref{result:performance of C(K(T))}) in Theorem~\ref{theorem:main theorem}. The results verify our main conclusion that, as $r$ and $\kappa$ increase, the performance gap between $\Kbf(T)$ and $\Kbf^*$ decreases exponentially. With a finite communication range $\kappa$, it is possible to design a scalable and decentralized gradient descent algorithm to obtain a near-optimal controller. Our future work will delve into scalable model-free algorithms, focusing on Monte-Carlo and Actor-Critic methods. The detailed discussion and more results can be found in Section~\uppercase\expandafter{\romannumeral5} in~\cite{supplementary}.



\begin{figure*}[t]
	\centering
	\begin{minipage}{0.49\linewidth}
		\centering
		\includegraphics[width=0.9\linewidth]{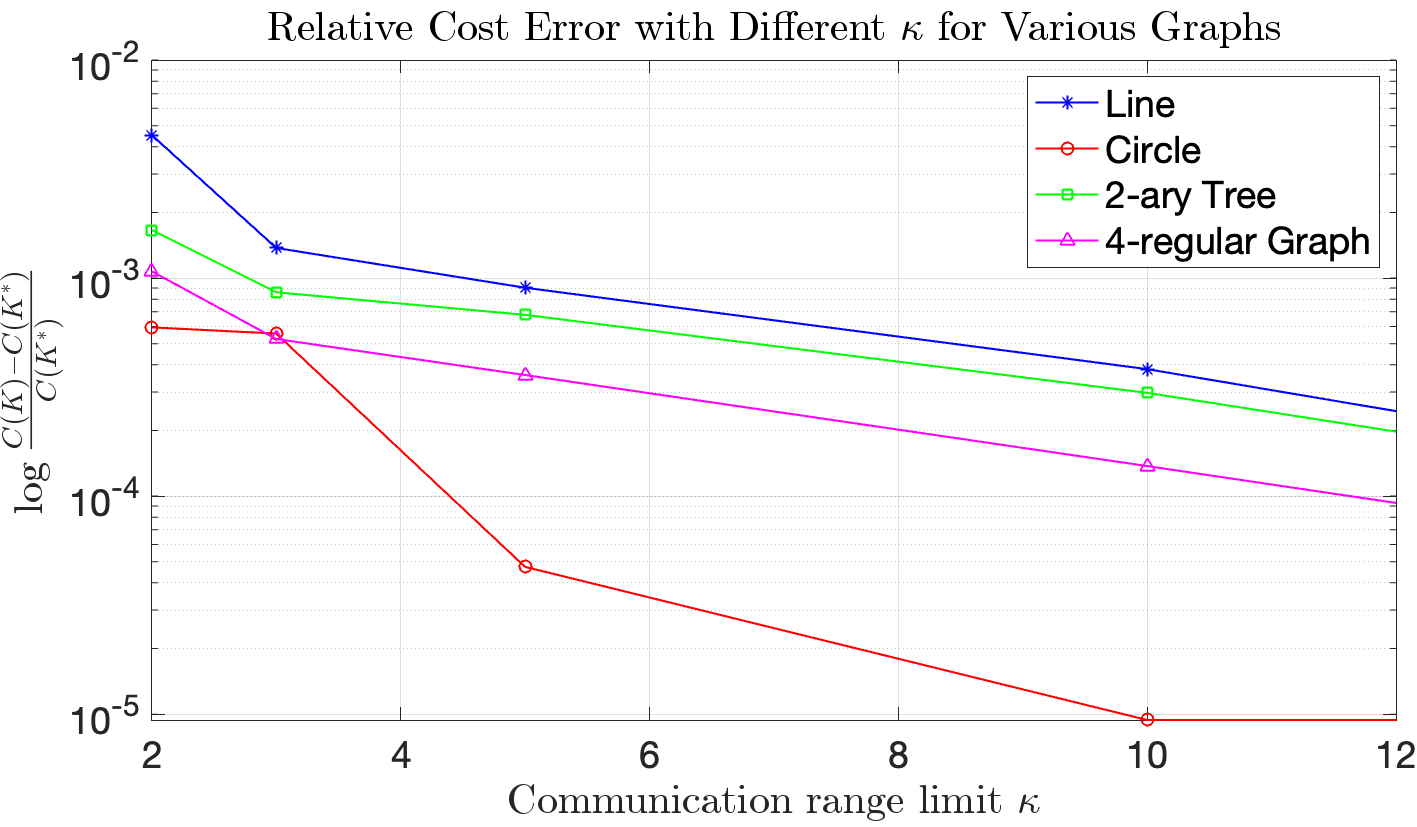}
	\end{minipage}
	\begin{minipage}{0.49\linewidth}
		\centering
		\includegraphics[width=0.9\linewidth]{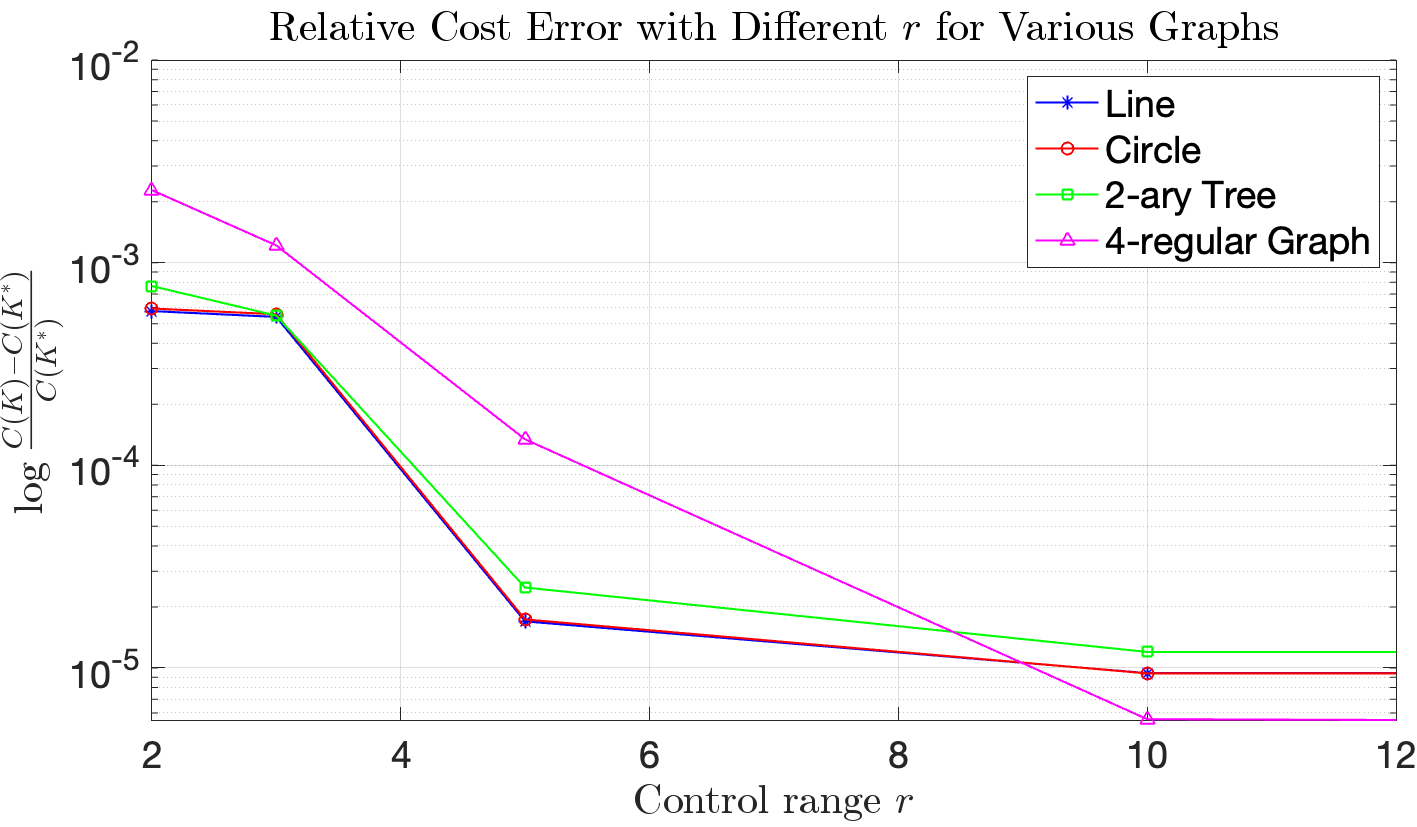}
	\end{minipage}
 	\caption{\small Relative performance gap compared to the optimal controller $\Kbf^*$ with different communication range limit $\kappa$ and different control range $r$. In semi-logarithmic plots with a logarithmic y-axis, the relative cost error curve is linear, aligning with theoretical results, and confirming the main conclusion that the performance gap between $\Kbf(T)$ and $\Kbf^*$ decreases exponentially as $r$ and $\kappa$ increase.}
        \label{fig: relative performance gap with kappa and r}
\end{figure*}

\subsection{Case Studies for Exponential Decay Property}
\label{subsec:Special Cases Study for Exponential Decay Property}
We give some typical examples and numerical results for Lemma~\ref{lem:conditionsForLQR2} to show that the Exponential Decay Property holds in these four representative graphs. Note that in Lemma~\ref{lem:conditionsForLQR2}, it requires that $| \cW_{i\rightarrow j}^t(r) | \leq C \cD^{t} {\rho}^{\kappa}$ and $\overline{[\Abf - \Bbf \Kbf]} \cdot \cD \leq 1$. For the first condition, we present the theoretical bound of $| \cW_{i\rightarrow j}^t(r) |$ and the exact value of $C$, $D$ and $\rho$ in TABLE~\ref{tab:theoretical_bound}. The second condition is an additional restriction for the stability of the whole system and the detailed discussion is deferred to Section~\uppercase\expandafter{\romannumeral6} in~\cite{supplementary}. 

\begin{remark}
We present the theoretical bound and the true number of walks in 4-regular grid in Fig. ~\ref{fig:number of walks in 4-regular graph} and more results are deferred to Section~\uppercase\expandafter{\romannumeral6} in~\cite{supplementary}. Note that the tighter theoretical bounds for the number of walks in these graphs could be expected, and are left as our future work. 
\end{remark}


\begin{figure*}[t]
	\centering
	\begin{minipage}{0.49\linewidth}
		\centering
		\includegraphics[width=0.9\linewidth]{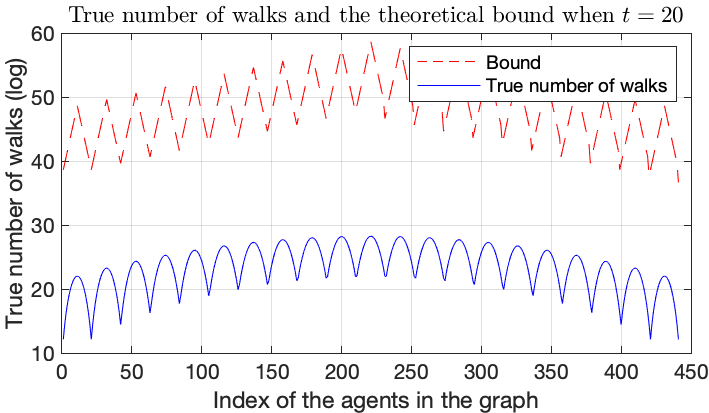}
	\end{minipage}
	\begin{minipage}{0.49\linewidth}
		\centering
		\includegraphics[width=0.9\linewidth]{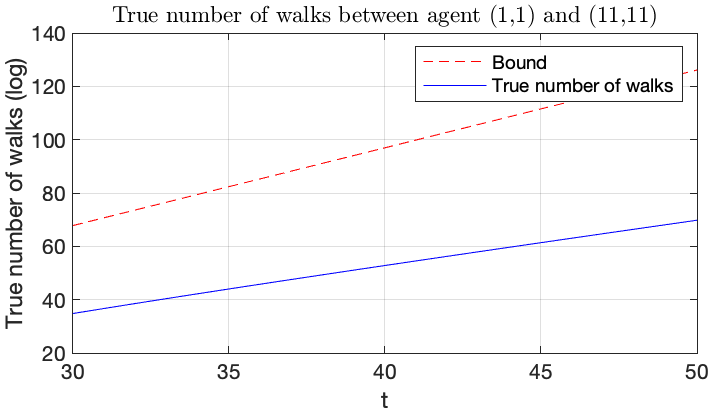}
	\end{minipage}
	\caption{\small The number of the true walks $| \cW_{i \rightarrow j} |$ and the theoretical bound in a $21 \times 21$ 4-regular grid. On the left, we fix $t=20$ and vary $\kappa$. In the right, We fix $\kappa=20$ and vary $t$.}
        \label{fig:number of walks in 4-regular graph}
\end{figure*}

\begin{table}[!t]
\caption{The theoretical bound of the number of the walks in four representative graphs: lines, circle, $f$-ary tree, and 4-regular grid.}
\centering
\renewcommand\arraystretch{2}
\label{tab:theoretical_bound}
\begin{tabular}{c|c|c|c|c}
\hline
Graph Type & Bound & $C$ & $D$ & $\rho$ \\
\hline
Line & $e \cdot [(\frac{3e}{2})^{\frac{3}{2}}]^t \cdot [e^{-\frac{1}{2}}]^\kappa$ & $e$ & $(\frac{3e}{2})^{\frac{3}{2}}$ & $e^{-\frac{1}{2}}$\\
\hline
Cycle & $\frac{e}{2n} (\frac{3}{2} e^{2})^t \cdot e^{-\kappa}$ & $\frac{e}{2n}$ & $\frac{3}{2} e^{2}$ & $e^{-\frac{1}{2}}$\\
\hline
$f$-ary Tree & $(2e^2 f^{\frac{1}{2}})^t (e^{-1}f^{-\frac{1}{2}})^{\kappa}$ & 1 & $2e^2 f^{\frac{1}{2}}$ & $e^{-1}f^{-\frac{1}{2}}$\\
\hline
4-regular Grid & $\frac{e}{2} (\frac{5e^2}{2})^t (e^{-1})^{\kappa}$ & $\frac{e}{2}$ & $\frac{5e^2}{2}$ & $e^{-1}$\\
\hline
\end{tabular}
\end{table}

\section{Conclusions and Future Works}
This paper has provided provable guarantees that the distributed policy gradient descent method converges to the near-optimal solution of networked LQR control problems.  
The \emph{near-optimality} represents the performance gap exponentially small in the communication range limit $\kappa$ and the control range $r$. We provided the conditions for localized gradient approximation and verify our proposition in some representative graphs. Additionally, we also showed that the controllers generated by the distributed policy gradient descent process can be guaranteed to stabilize the system. 

As the first attempt to providing  scalability and optimality in the distributed policy gradient method of networked LQR, our work 
has opened up several directions for  future research. We aim to develop zero-order and first-order optimization algorithms, incorporating theories like Monte-Carlo estimation as in~\cite{fazel2018global,li2021distributed,zhang2021derivative} and Actor-Critic as in~\cite{yang2019provably}. A focus includes finite-sample analysis during sample-based policy evaluation and examining the sample complexity in model-free contexts, as seen in~\cite{fazel2018global,cen2022fast}. Additionally, we seek a deeper understanding and tighter analysis of localized gradient approximation in networked LQ control, potentially integrating recent graph theory insights.

\section*{Acknowledgement} 
The authors would like to thank Prof. Kaiqing Zhang for helpful discussions and feedback.

\appendices
\section{Helper Definitions and Lemmas}
\label{appen:subsec:helper definitions and lemmas}
In this appendix, we provide a series of auxiliary definitions and lemmas aimed at enhancing the reader's comprehension. We recall the state transition dynamics from the global perspective:
\mequa
\xbf(t+1)=(\Abf-\Bbf \Kbf)\xbf(t)+\mathbf{\epsilon(t)},
\mequa
where $\mathbf{\epsilon(t)} = \wbf(t) + \sigma_0 \Bbf \eta(t) \sim N(\mathbf{0}, \bPsif)$ and we define $\bPsif = \bPhif+\sigma_0^2 \cdot \Bbf \Bbf^\Tb$ to simplify the notation.



Subsequently, we introduce a set of definitions to elucidate the relationship between the matrix structure and the inherent architecture of the network system. First, we recall a term coined to characterize the sparse nature of the matrices within this specific context.

Next, we will demonstrate the sparsity of the structure from a local perspective. To show how two agents interact from agent $i$'s point of view, we introduce a new definition as follows:
\begin{definition} \label{appen:defin:local spatially exponential decaying (L-SED)}
[Local Spatially Exponential Decaying (L-SED)] 
A matrix $\Xbf_i \in \RR^{n_1 \times n_2}$ related to agent $i$ is $(c_i, \gamma_i)$-local spatially exponential decaying (SED) if,
$
\Vert [\Xbf_i]_{mn} \Vert \leq c_i {\gamma_i}^{\dist(m,i)+\dist(i,n)},
$
where $0<\gamma_i<1$, $c_i>0$.
\end{definition}

Next, we will highlight certain properties of matrices that are closed under the SED condition.

\begin{lemma}
\label{appen:lem:SEDplusSEDandSEDtimesSED}
Suppose two $n$-by-$n$ block matrices $\Xbf$ and $\Ybf$ are square matrices of the same dimension, and they are $(x,\gamma)$-SED and $(y,\gamma)$-SED respectively, then the matrix $\Xbf+\Ybf$ is $(x+y, \gamma)$-SED and the matrix $\Xbf \cdot \Ybf$ is $(nxy, \gamma)$-SED.
\end{lemma}

\begin{lemma}
\label{appen:lem:SEDplusConnandSEDtimesConn}
Suppose $\Xbf$ is $(x,\gamma)$-SED, $\Ybf \in \cM^{\kappa}$ and $\max_{ij} \Vert [\Ybf]_{ij} \Vert = \bar{y}$, the $\Xbf \Ybf$ is $(nx\bar{y} e^{\gamma \kappa}, \gamma)$-SED, $\Xbf+\Ybf$ is $(x+\frac{\hat{y}}{e^{-\gamma \kappa}}, \gamma)$-SED.
\end{lemma}

\begin{lemma}
\label{appen:lem:ConnplusConnandConntimesConn}
Suppose $X\in \cM^{\kappa_x}$ and $\max_{ij} \Vert [\Xbf]_{ij} \Vert = \bar{x}$, $Y\in \cM^{\kappa_y}$ and $\max_{ij} \Vert [\Ybf]_{ij} \Vert = \bar{y}$, then $\Xbf \Ybf\in \cM^{\kappa_x + \kappa_y}$.
\end{lemma}

Next, we introduce a set of definitions from graph theory, aiming to clearly outline the necessary conditions for the Exponential Decay Property, as specified in Lemma~\ref{lem:conditionsForLQR2}.
\begin{definition}
$i\rightarrow n_1 \rightarrow n_2 \rightarrow ... \rightarrow n_{t-1} \rightarrow j$ is defined as a walk of length $t$ from $i$ to $j$, where $n_1$, $n_2$, ..., $n_{t-1} \in \cN$ 
\end{definition}

\begin{definition}
For convenience, we define an expanded connection graph $\cG(r)=(\cN, \cE(r))$  based on the underlying graph $\cG$. In $\cG(r)$, agent $i$ is connected to all its $r$-hop neighbors defined in $\cG$.
\end{definition}

\begin{definition}\label{appen:def:walk number}
We define $\cW_{i\rightarrow j}^t(r)$ as a set that contains all the walks from $i$ to $j$ with length $t$, in a defined graph $\cE(r)$.
\end{definition}

\section{Proof of Several Corollaries}
\subsection{Proof of Corollary~\ref{corollary: Kbf'' stabilizing}}
\label{appen:proof of corollary Kbf'' stabilizing}
By the equivalence of the Frobenius norm and the Euclidean norm, we have
\begin{equation}\label{appen:equ:K'-K''}
\begin{aligned}
& \Vert \Kbf''-\Kbf' \Vert  = \eta \Vert \cP_{\cM^r}(\nabla_{\Kbf} C(\Kbf)) - \hat{\hbf}(\Kbf)  \Vert \\
& \leq \eta \Vert \cP_{\cM^r}(\nabla_{\Kbf} C(\Kbf)) - \hat{\hbf}(\Kbf) \Vert_F \\
& = \eta \sum_{i=1}^n  \Vert \nabla_{\Kbf_i}C(\Kbf)-\hat{\hbf}_i(\Kbf) \Vert_F \\
& \leq
\eta \sqrt{d} \sum_{i=1}^n \Vert \nabla_{\Kbf_i} C(\Kbf)-\hat{\hbf}_i(\Kbf) \Vert \leq \eta \CC \sqrt{d} \sum_i L_i \rho^{\kappa+1}\\
& \leq \frac{ \sigma_1(\Qbf) \mu}{4 C(\Kbf) \Vert \Bbf \Vert (\Upsilon(C(\Kbf))+1)}
\end{aligned}
\end{equation}
where the second-to-last inequality results from Exponential Decay Property (Definition~\ref{def: Exponential Decay Property}). And we have
\begin{equation}
\label{appen:equ:bound A-BK'}
\begin{aligned}
& \Vert \Abf  -\Bbf \Kbf' \Vert = \Vert \Abf - \Bbf \Kbf + \Bbf(\Kbf-\Kbf')\Vert \\
\leq & \Vert \Abf-\Bbf \Kbf \Vert + \eta \Vert \Bbf \Vert \Vert \cP_{\cM^r}(\nabla C(\Kbf)) \Vert \\
\leq & \Vert \Abf - \Bbf \Kbf \Vert + \eta \sqrt{d} \Vert \Bbf \Vert \Vert \nabla C(\Kbf) \Vert \\
 \leq &\Vert \Abf \Vert + \Vert \Bbf \Vert \Vert \Kbf \Vert  \\
 & + \eta \sqrt{d} \Vert \Bbf \Vert \frac{C(\Kbf)}{\sigma_{1}(\Qbf)} \sqrt{\frac{\Vert \Rbf+\Bbf^\Tb \Pbf_\Kbf \Bbf \Vert (C(\Kbf)-C(\Kbf^*))}{\sigma_{1}(\bPsif)}} \\
 \leq &\Vert \Abf \Vert + \sqrt{d} \Vert \Bbf \Vert \frac{1}{\sigma_{1}(\Rbf)} \Big( \sqrt{ \frac{\Vert \Rbf+\Bbf^\Tb \Pbf_\Kbf \Bbf \Vert (C(\Kbf)-C(\Kbf^*))}{\sigma_{1}(\bPsif)}} \\
 & + \Vert \Bbf^\Tb \Pbf_\Kbf \Abf \Vert \Big)\\
& + \Vert \Bbf \Vert \frac{C(\Kbf)}{\sigma_{1}(\Qbf)} \sqrt{\frac{\Vert \Rbf+\Bbf^\Tb \Pbf_\Kbf \Bbf \Vert (C(\Kbf)-C(\Kbf^*))}{\sigma_{1}(\bPsif)}} \\
\coloneqq & \Upsilon(C(\Kbf)),
\end{aligned}
\end{equation}
where the bounds of $\Vert \Kbf \Vert$, $\Vert \nabla_{\Kbf} C(\Kbf) \Vert$, $\Vert \Pbf_\Kbf \Vert$ are given in~\cite{fazel2018global} (Lemma 13 $\&$ 25). So that
\mequa
\Vert \Kbf''-\Kbf' \Vert \leq \frac{ \sigma_1(\Qbf) \mu}{4 C(\Kbf) \Vert \Bbf \Vert (\Vert \Abf-\Bbf \Kbf'\Vert+1)}
\mequa
Based on Lemma 16 in \cite{fazel2018global}, and given the established stabilization of $\Kbf'$, it follows that $\Kbf''$ is also stabilizing. It leads to the third term in (\ref{condition:stepsizecondition}). 

\subsection{Proof of Corollary~\ref{corollary:diff between C(K'') and C(K')}}
\label{appen: proof diff between C(K'') and C(K')}
Note that $C(\Kbf)=\EE_{\xbf\sim N(\mathbf{0}, \bPsif)}\xbf^\Tb \Pbf_{\Kbf}\xbf + \sigma_0^2 \tr(\Rbf)$ (\cite{yang2019provably}, Proposition 3.1). For any $\Kbf$, the noise remains the same and i.i.d. When we calculate the difference between two objective functions, the effect of noise can be canceled out. So we can define a new state dynamics without noise, where $\xbf_t''=(\Abf - \Bbf \Kbf'')^t \xbf$ for all $t\geq 0$ and $\mathbf{\Xi}_{\Kbf''}=\EE_{\xbf \sim N(\mathbf{0}, \bPsif)}[\sum_{t\geq0} \xbf_t'' (\xbf_t'')^\Tb]$. 
\begin{equation}
\label{equ: bound on CKbf}
\begin{aligned} 
& C(\Kbf'')-C(\Kbf') = \EE_{\xbf\sim N(\mathbf{0}, \bPsif)}[\xbf^\Tb (\Pbf_{\Kbf''}-\Pbf_{\Kbf'})\xbf] \\
=& \EE_{\xbf\sim N(\mathbf{0}, \bPsif)} \sum_{t\geq0} A_{\Kbf',\Kbf''}(\xbf_t'')\\
\leq &  2\sqrt{d} \Vert \mathbf{\Xi}_\Kbf'' - \mathbf{\Xi}_\Kbf' \Vert \Vert \Ebf_{\Kbf'} \Vert \Vert \Kbf''-\Kbf' \Vert \\
& + 2\sqrt{d} \Vert \mathbf{\Xi}_\Kbf' \Vert \Vert \Ebf_{\Kbf'} \Vert \Vert \Kbf''-\Kbf' \Vert \\
 & + \sqrt{d} \Vert \mathbf{\Xi}_\Kbf''-\mathbf{\Xi}_\Kbf' \Vert \Vert \Rbf+\Bbf^\Tb \Pbf_\Kbf' \Bbf \Vert (\Vert \Kbf''-\Kbf' \Vert)^2 \\
 & + \sqrt{d} \Vert \mathbf{\Xi}_\Kbf' \Vert \Vert \Rbf+\Bbf^\Tb \Pbf_\Kbf' \Bbf \Vert (\Vert \Kbf''-\Kbf' \Vert)^2,
\end{aligned} 
\end{equation}
where $\Ebf_\Kbf=(\Rbf+\Bbf^\Tb \Pbf_\Kbf \Bbf)\Kbf-\Bbf^\Tb \Pbf_\Kbf \Abf$, and $A_{\Kbf',\Kbf''}$ is the advantage function defined as $A_{\Kbf,\Kbf'}(\xbf)=2\xbf^\Tb (\Kbf'-\Kbf)^\Tb \Ebf_\Kbf \xbf + \xbf^\Tb (\Kbf'-\Kbf)^\Tb (\Rbf+\Bbf^\Tb \Pbf_\Kbf \Bbf)(\Kbf'-\Kbf)\xbf$.The last inequality follows the cost difference lemma in \cite{fazel2018global}. Exponential Decay Property is used to bound $\Vert \Kbf''-\Kbf' \Vert$ and $\Vert \mathbf{\Xi}_{\Kbf''}-\mathbf{\Xi}_{\Kbf'} \Vert$ by term $\eta \rho^{\kappa+1}$. $\Vert \Kbf'' - \Kbf' \Vert$ is bounded as~\ref{appen:equ:K'-K''} and $\Vert \mathbf{\Xi}_{\Kbf''}-\mathbf{\Xi}_{\Kbf'} \Vert$ can be bounded in the same way based on Lemma 16 in \cite{fazel2018global},
\mequa
\Vert \mathbf{\Xi}_{\Kbf'} - \mathbf{\Xi}_\Kbf \Vert \leq 4(\frac{C(\Kbf)}{\sigma_1(\Qbf)})^2 \frac{\Vert \Bbf \Vert (\Vert \Abf - \Bbf \Kbf   \Vert + 1)}{\mu} \Vert \Kbf' - \Kbf \Vert.
\mequa

Besides, we have two additional terms to bound with terms related to $\Kbf$:  $\Vert \Ebf_{\Kbf'} \Vert$, $\Vert \mathbf{\Xi}_{\Kbf'} \Vert$. Here we present the results directly.
\#\label{appen:equ:bound of XiK'}
\Vert \mathbf{\Xi}_{\Kbf'} \Vert \leq  \frac{C(\Kbf')}{\sigma_1(\Qbf)} \leq \frac{C(\Kbf)}{\sigma_1(\Qbf)} .
\#

\begin{equation}
\label{appen:equ:EK' bound}
\begin{aligned}
\Vert \Ebf_{\Kbf'} \Vert  \leq \sqrt{\frac{(C(\Kbf)-C(\Kbf^*))\big\Vert \Rbf+\Bbf^\Tb \frac{C(\Kbf)}{\sigma_{1}(\bPsif)} \Bbf \big\Vert}{\sigma_{1}(\bPsif)}},
\end{aligned}
\end{equation}
Then, all the terms in~\ref{equ: bound on CKbf} can be bounded by $\eta \rho^{\kappa+1}$, terms related to $C(\Kbf)$ and other system parameters. 
The detailed discussion and the explicit form of the polynomial are deferred to the Part D, Section~\uppercase\expandafter{\romannumeral2} in~\cite{supplementary}.

\subsection{Iteration Process of the Distributed Policy Gradient Descent}
\label{appen:subsec:iteration results}
Inductively, we have,
\mequa
& C(\Kbf(T))-C(\Kbf^*) \\
\leq &\big(1- \eta \frac{\mu^2 \sigma_1(\Rbf)}{\Vert \mathbf{\Xi}_{\Kbf^*} \Vert}\big) \big[C(\Kbf(T-1)) - C(\Kbf^*)\big] \\
& + M_1 F_1(C(\Kbf(0))) \rho^{\kappa+1} + M_2 F_2(C(\Kbf(0)))\rho^r\\
\leq & \big(1- \eta \frac{\mu^2 \sigma_1(\Rbf)}{\Vert \mathbf{\Xi}_{\Kbf^*} \Vert}\big)^T \big[C(\Kbf(0)) - C(\Kbf^*)\big] \\
& + \prod_{\tau = 0}^{T-1}(1- \eta \frac{\mu^2 \sigma_1(\Rbf)}{\Vert \mathbf{\Xi}_{\Kbf^*} \Vert})^{\tau}[M_1 F_1(C(\Kbf(0))) \rho^{\kappa+1} \\
& + M_2 F_2(C(\Kbf(0)))\rho^r] \\
\leq & \big(1- \eta \frac{\mu^2 \sigma_1(\Rbf)}{\Vert \mathbf{\Xi}_{\Kbf^*} \Vert}\big)^T (C(\Kbf(0)) - C(\Kbf^*))\\
& + \frac{\Vert \mathbf{\Xi}_{\Kbf^*} \Vert}{\eta \mu^2 \sigma_1(\Rbf)} [M_1 F_1(C(\Kbf(0))) \rho^{\kappa+1} + M_2 F_2(C(\Kbf(0)))\rho^r].
\mequa

Provided
\mequa
T \geq \frac{\Vert \mathbf{\Xi}_{\Kbf^*} \Vert}{\eta \mu^2 \sigma_1(\Rbf)} \log \frac{C(\Kbf(0))-C(\Kbf^*)}{\epsilon},
\mequa
then we have
\mequa
& C(\Kbf(T))-C(\Kbf^*) \leq \epsilon\\
& + \frac{\Vert \mathbf{\Xi}_{\Kbf^*} \Vert}{\eta \mu^2 \sigma_1(\Rbf)} \big[M_1 F_1(C(\Kbf(0))) \rho^{\kappa+1}  + M_2 F_2(C(\Kbf(0))) \rho^r\big].
\mequa

\bibliographystyle{IEEEtran}

\bibliography{IEEEabrv,StringDefinitions,SGroupDefinition,SGroup, SGroup_Yuzi}

\vfill

\end{document}